\documentclass
[
    oneside,                 
    openright,               
    cleardoublepage = empty, 
    fontsize = 12 pt,        
    american,                
    captions = tableheading, 
    numbers = noenddot,      
    footheight = 35 pt,      
]
{scrbook}

\newif\ifprintVersion   
\newif\ifprofessionalPrint 
\newif\iffancyTheorems  
\newif\ifboldNumberSets 
\newif\ifbachelorThesis 

\printVersionfalse
\professionalPrintfalse
\fancyTheoremstrue
\boldNumberSetstrue
\bachelorThesistrue

\newcommand*{\printTitle}{}
\newcommand*{\printGermanTitle}{}
\newcommand*{\myTitle}[2]{\renewcommand*{\printTitle}{#1}\renewcommand*{\printGermanTitle}{#2}}
\newcommand*{\printTitleBold}{\textbf{\printTitle}}

\newcommand*{\printAuthor}{}
\newcommand*{\myName}[1]{\renewcommand*{\printAuthor}{#1}}

\newcommand*{\printProgram}{}
\newcommand*{\myProgram}[1]{\renewcommand*{\printProgram}{#1}}

\newcommand*{\printDateReceived}{}
\newcommand*{\dateOfHandingIn}[1]{\renewcommand*{\printDateReceived}{#1}}

\newcommand*{\printSubject}{}
\newcommand*{\mySubject}[1]{\renewcommand*{\printSubject}{#1}}

\newcommand*{\printKeywords}{}
\newcommand*{\myKeywords}[1]{\renewcommand*{\printKeywords}{#1}}

\newcommand*{\printNameOfSupervisor}{}
\newcommand*{\nameOfMySupervisor}[1]{\renewcommand*{\printNameOfSupervisor}{#1}}

\newcommand*{\printAdditionalExaminers}{}
\newcommand*{\additionalExaminers}[1]{\renewcommand*{\printAdditionalExaminers}{#1}}

\newlength{\extraborderlength}
\newcommand*{\extraBorder}[1]{\setlength{\extraborderlength}{#1}}

\newlength{\mybindingcorrection}
\newcommand*{\bindingCorrection}[1]{\setlength{\mybindingcorrection}{#1}} 

    \extraBorder{3 mm}

    \bindingCorrection{6 mm}

    \myTitle{Analysis of a Random Local Search Algorithm for Dominating Set}{Analyse einer randomisierten lokalen Suche für Dominating Set}

    \myName{Hendrik Higl}

    \myProgram{IT Systems Engineering}

    \dateOfHandingIn{30. Oktober 2025}

    \nameOfMySupervisor{Prof.\,Dr. Tobias Friedrich}

    \additionalExaminers{Dr. Timo Kötzing\newline Jurek Sander}

    \mySubject{A bachelor thesis for the degree of IT Systems Engineering at the university of Potsdam. We analyze the run time of a random local search algorithm for the dominating set problem on cycle graphs.}

    \myKeywords{bachelor thesis | random local search | run time analysis | dominating set | cycle graph | markov chain}

\usepackage[utf8]{inputenc} 
\usepackage[T1]{fontenc}    
\usepackage
[
    ngerman,         
    main = american, 
]
{babel}                     

\input glyphtounicode
\usepackage{calc} 

\newlength{\myparindent}
\newlength{\myparskip}
\setfootnoterule{0 cm}

\deffootnote[1.2 em]{1.2 em}{0 em}{\makebox[1.4 em][l]{\textbf{\thefootnotemark}}}

\makeatletter%
    \@removefromreset{footnote}{chapter}%
\makeatother

\usepackage[dvipsnames]{xcolor} 

\definecolor{stroke1}{HTML}{2574A9} 

\colorlet{captionlabel}{black}
\colorlet{footerpagenr}{black}
\colorlet{footerchapter}{stroke1}
\colorlet{footerchaptername}{black}
\colorlet{footersection}{stroke1}
\colorlet{footersectionname}{black}
\colorlet{chapternumber}{stroke1}

\newlength{\mypaperwidth}
\newlength{\mypaperheight}
\newlength{\mybodywidth}
\newlength{\mybodyheight}
\newlength{\myoutermargin}
\ifprintVersion
    \ifprofessionalPrint
    \else
    \fi
\else
\fi

\newlength{\mytopmargin}
\ifprintVersion
    \ifprofessionalPrint
    \fi
\fi

\newlength{\myinnermargin}
\ifprintVersion
    \ifprofessionalPrint
    \fi
\fi

\newlength{\mybottommargin}
\ifprintVersion
    \ifprofessionalPrint
    \fi
\fi

\newcommand{\goldenratio}{1.618}

\newlength{\myheadsep} 
\ifprintVersion
    \ifprofessionalPrint
    \fi
\fi

\newlength{\myfootskip} 
\ifprintVersion
    \ifprofessionalPrint
    \fi
\fi

\newlength{\mymargininnersep} 
\newlength{\mymarginoutersep} 
\ifprintVersion
    \ifprofessionalPrint
    \fi
\fi

\newlength{\mymarginwidth} 
\newlength{\mymarginwidthwithinnersep} 
\usepackage
[
    \ifprintVersion
        \ifprofessionalPrint
            paperwidth = \mypaperwidth + 2\extraborderlength + \mybindingcorrection,
            paperheight = \mypaperheight + 2\extraborderlength,
        \else
            paperwidth = \mypaperwidth,
            paperheight = \mypaperheight,
        \fi
    \else
        paperwidth = \mypaperwidth,
        paperheight = \mypaperheight,
    \fi
    textwidth = \mybodywidth,
    textheight = \mybodyheight,
    outer = \myoutermargin,
    top = \mytopmargin,
    headsep = \myheadsep,
    footskip = \myfootskip,
    marginparsep = \mymargininnersep,
    marginparwidth = \mymarginwidth,
]
{geometry} 

\usepackage
[
]
{scrlayer-scrpage} 

\clearpairofpagestyles

\KOMAoptions
{%
    headwidth = \textwidth + \mymarginwidthwithinnersep,%
    footwidth = \myoutermargin : \textwidth,%
}

\automark[chapter]{chapter}
\automark*[section]{}

\lehead%
{%
    \begin{minipage}[b]{\mymarginwidth}%
        \small\raggedleft\normalfont\textsf{\textbf{\color{footerchapter}\chaptername\ \thechapter}}
    \end{minipage}
}
\cehead{\hspace*{\mymarginwidthwithinnersep}\parbox{\textwidth}{\raggedright\leftmark}}

\rohead%
{%
    \Ifstr{\rightmark}{\leftmark}%
    {%
        \begin{minipage}[b]{\mymarginwidth}%
            \small\raggedright\normalfont\textsf{\textbf{\color{footersection}Chapter\ \thechapter}}%
        \end{minipage}%
    }%
    {%
        \begin{minipage}[b]{\mymarginwidth}%
            \small\raggedright\normalfont\textsf{\textbf{\color{footersection}Section\ \thesection}}%
        \end{minipage}%
    }%
}
\cohead{\hspace*{-\mymarginwidthwithinnersep}\parbox{\textwidth}{\raggedleft\rightmark}}

\lefoot*%
{%
    \vspace*{1 ex}%
    {\color{stroke1}\rule{\myoutermargin - \mymargininnersep}{0.5 mm}}\\
    \begin{minipage}[b]{\myoutermargin - \mymargininnersep}%
        \raggedleft\normalfont\color{footerpagenr}\textbf{\thepage}%
    \end{minipage}%
}
\rofoot*%
{%
    {\color{stroke1}\rule{\myoutermargin - \mymargininnersep}{0.5 mm}}\\
    \begin{minipage}[b]{\myoutermargin - \mymargininnersep}%
        \raggedright\normalfont\color{footerpagenr}\textbf{\thepage}%
    \end{minipage}%
}

\usepackage{caption}
\usepackage{tikz} 

\newlength{\mytmpa}
\newlength{\mytmpb}

\renewcommand*{\partlineswithprefixformat}[3]%
{%
    #2
    \thispagestyle{empty}
    \setlength{\mytmpa}{0.618\mypaperwidth}%
    \setlength{\mytmpb}{0.382\mypaperheight}%
    \ifprintVersion
        \ifprofessionalPrint
            \setlength{\mytmpa}{0.618\mypaperwidth + \mybindingcorrection + \extraborderlength}%
            \setlength{\mytmpb}{0.382\mypaperheight + \extraborderlength}%
        \fi
    \fi
    \begin{tikzpicture}[overlay, remember picture]%
        \node [inner sep = 0, outer sep = 0, anchor = north] at (current page.north west)%
        {%
            \begin{tikzpicture}[overlay, remember picture]%
            \draw[color = stroke1, line width = 0.7 mm] (\mytmpa, 0) -- (\mytmpa, -\mytmpb);%
            \end{tikzpicture}%
        };%
        \node (align) [align = right, below = \mytmpb - 2 ex, inner sep = 0, outer sep = 0, anchor = north west] at (current page.north west)%
        {%
            \hspace{\mytmpa}\hspace{0.5 em}\partname\ \thepart\\[1 ex]
            \color{stroke1}#3%
        };%
    \end{tikzpicture}%
}
\RedeclareSectionCommand%
[%
    font = \normalfont\Huge\sffamily,
    prefixfont = \normalfont\Huge\sffamily,
]
{part}

\usepackage{etoolbox}

\newbool{chapterHasANumber}
\newbool{chapterHasAStar}
\renewcommand*{\chapterlinesformat}[3]%
{%
    \Ifnumbered{#1}{\setbool{chapterHasANumber}{true}}{\setbool{chapterHasANumber}{false}}%
    \Ifstr{#2}{}{\setbool{chapterHasAStar}{true}}{\setbool{chapterHasAStar}{false}}%
    \ifboolexpr{bool{chapterHasANumber} and not bool{chapterHasAStar}}%
    {%
        \begin{tikzpicture}[overlay, remember picture]%
            \node [right = \myinnermargin, below = \mytopmargin, inner sep = 0, outer sep = 0, anchor = north west] (numbernode) at (current page.north west)%
            {%
                \hspace{\myinnermargin}%
                \sffamily\fontsize{60}{60}\selectfont%
                \color{chapternumber}%
                \thechapter%
            };%
            \node [inner sep = 0, outer sep = 0, anchor = north west] at (numbernode.south west)%
            {%
                \begin{tikzpicture}[overlay, remember picture]%
                    \draw[color = stroke1, line width = 0.7 mm] (\myinnermargin, -1 ex) -- (\paperwidth, -1 ex);%
                \end{tikzpicture}%
            };%
            \node (align) [text width = \textwidth - 2 cm, align = right, right = \myinnermargin + \mybodywidth, inner sep = 0, outer sep = 0, anchor = east] at (numbernode.west)%
            {%
                #3%
            };%
        \end{tikzpicture}%
    }%
    {%
        \begin{tikzpicture}[overlay, remember picture]%
            \node [right = \myinnermargin, below = \mytopmargin, inner sep = 0, outer sep = 0, anchor = north west] (numbernode) at (current page.north west)%
            {%
                \hspace{\myinnermargin}%
                \sffamily\fontsize{60}{60}\selectfont%
                \color{white}%
                \thechapter%
            };%
            \node [inner sep = 0, outer sep = 0, anchor = north west] at (numbernode.south west)%
            {%
                \begin{tikzpicture}[overlay, remember picture]%
                    \draw[color = stroke1, line width = 0.7 mm] (\myinnermargin, -1 ex) -- (\paperwidth, -1 ex);%
                \end{tikzpicture}%
            };%
            \node (align) [align = left, right = \myinnermargin, inner sep = 0, outer sep = 0, anchor = south west] at (numbernode.south west)%
            {%
                #3%
            };%
        \end{tikzpicture}%
    }%
}
\RedeclareSectionCommand%
[%
    font = \color{stroke1}\normalfont\huge\sffamily,
    afterskip = 20 pt,
]
{chapter}

\BeforeStartingTOC[toc]{\pagestyle{plain}}
\AfterStartingTOC{\thispagestyle{plain}}                        
\usepackage
[
    sortcites,              
    style = alphabetic,     
    giveninits=true,
    defernumbers,           
    safeinputenc,           
    backref = true,         
    backrefstyle = three,   
    hyperref = true,        
    maxbibnames = 99,       
    maxcitenames = 2,       
]
{biblatex} 

\renewbibmacro{in:}%
{%
    \ifentrytype{article}{}{\printtext{\bibstring{in}\intitlepunct}}%
}

\renewbibmacro*{volume+number+eid}%
{%
    \printfield{volume}%
    \iffieldundef{number}{}{\addcolon}%
    \printfield{number}%
    \setunit*{\addcomma\space}%
    \printfield{eid}%
}

\DefineBibliographyStrings{english}%
{%
    backrefpage  = {\lowercase{s}ee page}, 
    backrefpages = {\lowercase{s}ee pages} 
}

\DeclareFieldFormat[article]{title}{\textbf{\color{stroke1}#1}}
\DeclareFieldFormat[inproceedings]{title}{\textbf{\color{stroke1}#1}}
\DeclareFieldFormat[thesis]{title}{\textbf{\color{stroke1}#1}}
\DeclareFieldFormat[book]{title}{\textbf{\color{stroke1}#1}}
\DeclareFieldFormat[unpublished]{title}{\textbf{\color{stroke1}#1}}
\DeclareFieldFormat[report]{title}{\textbf{\color{stroke1}#1}}
\DeclareFieldFormat[inbook]{chapter}{\textbf{\color{stroke1}#1}}
\DeclareFieldFormat[inbook]{title}{#1}
\DeclareFieldFormat{pages}{#1}

\newtoggle{authorend}
\togglefalse{authorend}

\DeclareBibliographyDriver{article}%
{%
  \usebibmacro{bibindex}%
  \usebibmacro{begentry}%
  \iftoggle{authorend}{}{\usebibmacro{author/translator+others}}%
  \setunit{\labelnamepunct}\newblock
  \usebibmacro{title}%
  \newunit
  \newunit\newblock
  \usebibmacro{byauthor}%
  \newunit\newblock
  \usebibmacro{bytranslator+others}%
  \newunit\newblock
  \printfield{version}%
  \newunit\newblock
  \usebibmacro{in:}%
  \usebibmacro{journal+issuetitle}%
  \newunit
  \usebibmacro{byeditor+others}%
  \newunit
  \usebibmacro{note+pages}%
  \newunit\newblock
  \iftoggle{bbx:isbn}
  {\printfield{issn}}
  {}%
  \newunit\newblock
  \usebibmacro{doi+eprint+url}%
  \newunit\newblock
  \usebibmacro{addendum+pubstate}%
  \setunit{\bibpagerefpunct}\newblock
  \usebibmacro{pageref}%
  \newunit\newblock
  \iftoggle{bbx:related}
  {\usebibmacro{related:init}%
    \usebibmacro{related}}
  {}%
  \usebibmacro{finentry}%
  \iftoggle{authorend}{\usebibmacro{author/translator+others}}{}%
}

\DeclareBibliographyDriver{inbook}%
{%
  \usebibmacro{bibindex}%
  \usebibmacro{begentry}%
  \iftoggle{authorend}{}{\usebibmacro{author/translator+others}}%
  \setunit{\labelnamepunct}\newblock
  \usebibmacro{chapter+pages}%
  \newunit
  \newunit\newblock
  \usebibmacro{byauthor}%
  \newunit\newblock
  \usebibmacro{in:}%
  \usebibmacro{bybookauthor}%
  \newunit\newblock
  \usebibmacro{maintitle+booktitle}%
  \newunit\newblock
  \usebibmacro{byeditor+others}%
  \newunit\newblock
  \printfield{edition}%
  \newunit
  \iffieldundef{maintitle}
  {\printfield{volume}%
    \printfield{part}}
  {}%
  \newunit
  \printfield{volumes}%
  \newunit\newblock
  \usebibmacro{series+number}%
  \newunit\newblock
  \printfield{note}%
  \newunit\newblock
  \usebibmacro{publisher+location+date}%
  \newunit\newblock
  \newunit\newblock
  \iftoggle{bbx:isbn}
  {\printfield{isbn}}
  {}%
  \newunit\newblock
  \usebibmacro{doi+eprint+url}%
  \newunit\newblock
  \usebibmacro{addendum+pubstate}%
  \setunit{\bibpagerefpunct}\newblock
  \usebibmacro{pageref}%
  \newunit\newblock
  \iftoggle{bbx:related}
  {\usebibmacro{related:init}%
    \usebibmacro{related}}
  {}%
  \usebibmacro{finentry}%
  \iftoggle{authorend}{\usebibmacro{author/translator+others}}{}%
}

\DeclareBibliographyDriver{inproceedings}%
{%
  \usebibmacro{bibindex}%
  \usebibmacro{begentry}%
  \iftoggle{authorend}{}{\usebibmacro{author/translator+others}}%
  \setunit{\labelnamepunct}\newblock
  \usebibmacro{title}%
  \newunit
  \newunit\newblock
  \usebibmacro{byauthor}%
  \newunit\newblock
  \usebibmacro{in:}%
  \usebibmacro{maintitle+booktitle}%
  \newunit\newblock
  \usebibmacro{event+venue+date}%
  \newunit\newblock
  \usebibmacro{byeditor+others}%
  \newunit\newblock
  \iffieldundef{maintitle}
  {\printfield{volume}%
    \printfield{part}}
  {}%
  \newunit
  \printfield{volumes}%
  \newunit\newblock
  \usebibmacro{series+number}%
  \newunit\newblock
  \printfield{note}%
  \newunit\newblock
  \printlist{organization}%
  \newunit
  \usebibmacro{publisher+location+date}%
  \newunit\newblock
  \usebibmacro{chapter+pages}%
  \newunit\newblock
  \iftoggle{bbx:isbn}
  {\printfield{isbn}}
  {}%
  \newunit\newblock
  \usebibmacro{doi+eprint+url}%
  \newunit\newblock
  \usebibmacro{addendum+pubstate}%
  \setunit{\bibpagerefpunct}\newblock
  \usebibmacro{pageref}%
  \newunit\newblock
  \iftoggle{bbx:related}
  {\usebibmacro{related:init}%
    \usebibmacro{related}}
  {}%
  \usebibmacro{finentry}%
  \iftoggle{authorend}{\usebibmacro{author/translator+others}}{}%
}

\DeclareBibliographyDriver{thesis}%
{%
  \usebibmacro{bibindex}%
  \usebibmacro{begentry}%
  \iftoggle{authorend}{}{\usebibmacro{author}}%
  \setunit{\labelnamepunct}\newblock
  \usebibmacro{title}%
  \newunit
  \newunit\newblock
  \usebibmacro{byauthor}%
  \newunit\newblock
  \printfield{note}%
  \newunit\newblock
  \printfield{type}%
  \newunit
  \usebibmacro{institution+location+date}%
  \newunit\newblock
  \usebibmacro{chapter+pages}%
  \newunit
  \printfield{pagetotal}%
  \newunit\newblock
  \iftoggle{bbx:isbn}
  {\printfield{isbn}}
  {}%
  \newunit\newblock
  \usebibmacro{doi+eprint+url}%
  \newunit\newblock
  \usebibmacro{addendum+pubstate}%
  \setunit{\bibpagerefpunct}\newblock
  \usebibmacro{pageref}%
  \newunit\newblock
  \iftoggle{bbx:related}
  {\usebibmacro{related:init}%
    \usebibmacro{related}}
  {}%
  \usebibmacro{finentry}%
  \iftoggle{authorend}{\usebibmacro{author}}{}%
}

\providebool{bbx:subentry}
\newbibmacro*{citenum}%
{
  \printtext[bibhyperref]{
    \printfield{prefixnumber}%
    \printfield{labelnumber}%
    \ifbool{bbx:subentry}
    {\printfield{entrysetcount}}
    {}}%
}

\DeclareCiteCommand{\conline}[\mkbibbrackets]
{\usebibmacro{prenote}}
{\usebibmacro{citeindex}%
  \usebibmacro{citenum}}
{\multicitedelim}
{\usebibmacro{postnote}}       
\usepackage
[
    babel = true, 
]
{microtype}           
\usepackage{csquotes} 

\usepackage{amsmath}
\usepackage{amssymb}
\usepackage{amsthm}
\usepackage{thmtools}
\usepackage{mathtools}
\usepackage{thm-restate}
\usepackage{dsfont}        
\usepackage{braceMnSymbol} 

\usepackage
[
    ttscale = 0.85, 
]
{libertine} 
\usepackage
[
    libertine,    
    slantedGreek, 
    vvarbb,       
    libaltvw,     
]
{newtxmath} 
\usepackage{url} 
\usepackage{bm}  

\usepackage{graphicx} 
\usepackage
[
    subrefformat = simple, 
    labelformat = simple,  
]
{subcaption}         
\usepackage{wrapfig} 

\columnsep = \mymargininnersep

\usepackage{array}     
\usepackage{booktabs}  
\usepackage{longtable} 
\usepackage{pdflscape} 

\usepackage{enumitem} 

\usepackage
[
    ruled,         
    vlined,        
    linesnumbered, 
]
{algorithm2e} 

\usepackage{xspace}   
\usepackage
[
    shortcuts, 
]
{extdash}             
\usepackage{setspace} 

\usepackage{xparse}    
\usepackage{footnote}  
\usepackage{afterpage} 
\usepackage
[
    textsize = scriptsize, 
]
{todonotes}            

\usepackage{dirtytalk}
\usepackage{tabularx}

\usepackage
[
    bookmarks = true,                 
    bookmarksopen = false,            
    bookmarksnumbered = true,         
    pdfstartpage = 1,                 
    pdftitle = {{\printTitle}},       
    pdfauthor = {{\printAuthor}},     
    pdfsubject = {{\printSubject}},   
    pdfkeywords = {{\printKeywords}}, 
    breaklinks = true,                
    \ifprintVersion
        hidelinks,                    
    \else
    colorlinks = true,            
    allcolors = stroke1,          
    \fi
]
{hyperref} 

\usepackage
[
    noabbrev,   
    nameinlink, 
]
{cleveref} 
\newcommand*{\colloquialDegreeName}{Master}
\newcommand*{\colloquialDegreeNameLowercase}{master}

\newcommand*{\degreeAbbreviation}{M.}

\ifbachelorThesis
    \renewcommand*{\colloquialDegreeName}{Bachelor}
    \renewcommand*{\colloquialDegreeNameLowercase}{bachelor}
    \renewcommand*{\degreeAbbreviation}{B.}
\fi

\makeatletter
    \def\IfEmptyTF#1%
    {%
        \if\relax\detokenize{#1}\relax%
            \expandafter\@firstoftwo%
        \else%
            \expandafter\@secondoftwo%
        \fi%
    }
\makeatother

\NewDocumentCommand{\mathOrText}{m}
{%
    \ensuremath{#1}\xspace%
}

\let\originalleft\left
\let\originalright\right
\renewcommand{\left}{\mathopen{}\mathclose\bgroup\originalleft}
\renewcommand{\right}{\aftergroup\egroup\originalright}

\makeatletter
    \DeclareRobustCommand{\bfseries}%
    {%
        \not@math@alphabet\bfseries\mathbf%
        \fontseries\bfdefault\selectfont%
        \boldmath%
    }
\makeatother

\xspaceaddexceptions{]\}}

\allowdisplaybreaks

\crefname{ineq}{inequality}{inequalities}
\creflabelformat{ineq}{#2{\upshape(#1)}#3} 

\crefname{term}{term}{terms}
\creflabelformat{term}{#2{\upshape(#1)}#3}

\let\oldfootnote\footnote

\newlength{\spaceBeforeFootnote} 
\newlength{\spaceAfterFootnote}  

\RenewDocumentCommand{\footnote}{o o o m}%
{%
    \IfNoValueTF{#1}%
    {%
        \oldfootnote{#4}%
    }%
    {%
        \setlength{\spaceBeforeFootnote}{\IfEmptyTF{#1}{0}{#1} em}%
        \IfNoValueTF{#2}%
        {%
            \hspace*{\spaceBeforeFootnote}\oldfootnote{#4}%
        }%
        {%
            \setlength{\spaceAfterFootnote}{\IfEmptyTF{#2}{0}{#2} em}%
            \hspace*{\spaceBeforeFootnote}\IfNoValueTF{#3}{\oldfootnote{#4}}{\oldfootnote[#3]{#4}}\hspace*{\spaceAfterFootnote}%
        }%
    }%
}

\makesavenoteenv{figure}
\makesavenoteenv{table}
\makesavenoteenv{tabular}

\iffancyTheorems
    \declaretheoremstyle
    [
        spaceabove = \topsep,
        spacebelow = \topsep,
        headfont = \bfseries,
        headformat = \textcolor{stroke1}{$\blacktriangleright$} \NAME~\NUMBER \NOTE,
        notefont = \bfseries,
        notebraces = {(}{)},
        bodyfont = \normalfont,
        postheadspace = 0.5 em,
        qed = \textcolor{stroke1}{\bfseries$\blacktriangleleft$},
    ]
    {myTheoremStyle}
    
    \declaretheorem
    [
        style = myTheoremStyle,
        name = Conjecture,
        numberwithin = chapter,
    ]
    {conjecture}
    \declaretheorem
    [
        style = myTheoremStyle,
        name = Proposition,
        sharenumber = conjecture,
    ]
    {proposition}
    \declaretheorem
    [
        style = myTheoremStyle,
        name = Claim,
        sharenumber = conjecture,
    ]
    {claim}
    \declaretheorem
    [
        style = myTheoremStyle,
        name = Lemma,
        sharenumber = conjecture,
    ]
    {lemma}
    \declaretheorem
    [
        style = myTheoremStyle,
        name = Corollary,
        sharenumber = conjecture,
    ]
    {corollary}
    \declaretheorem
    [
        style = myTheoremStyle,
        name = Theorem,
        sharenumber = conjecture,
    ]
    {theorem}
    \declaretheorem
    [
        style = myTheoremStyle,
        name = Definition,
        sharenumber = conjecture,
    ]
    {definition}
    \declaretheorem
    [
        style = myTheoremStyle,
        name = Example,
        sharenumber = conjecture,
    ]
    {example}
    \declaretheorem
    [
        style = myTheoremStyle,
        name = Remark,
        sharenumber = conjecture,
    ]
    {remark}
\else
    \theoremstyle{plain}
    
    \newtheorem{conjecture}{Conjecture}[chapter]
    \newtheorem{proposition}[conjecture]{Proposition}
    
    \newtheorem{lemma}[conjecture]{Lemma}
    \newtheorem{corollary}[conjecture]{Corollary}
    \newtheorem{theorem}[conjecture]{Theorem}
    \newtheorem{definition}[conjecture]{Definition}

\fi

\NewDocumentCommand{\functionTemplate}{m m m m o}%
{%
    \IfNoValueTF{#5}%
    {%
        \mathOrText{#1\left#2{#4}\right#3}%
    }%
    {%
        \mathOrText{#1#5#2{#4}#5#3}%
    }%
}

\newcommand*{\leftBracketType}{(}
\newcommand*{\rightBracketType}{)}

\NewDocumentCommand{\createFunction}{m m o o}%
{%
    \renewcommand*{\leftBracketType}{\IfNoValueTF{#3}{(}{#3}}%
    \renewcommand*{\rightBracketType}{\IfNoValueTF{#4}{)}{#4}}%
    \NewDocumentCommand{#1}{o o}%
    {%
        \IfNoValueTF{##1}%
        {%
            \mathOrText{#2}%
        }%
        {%
            \functionTemplate{#2}{\leftBracketType}{\rightBracketType}{##1}[##2]%
        }%
    }%
}

\DeclareDocumentCommand{\probabilisticFunctionTemplate}{m m O{} o}
{%
    \functionTemplate{#1}%
    {\lbrack}%
    {\rbrack}%
    {#2\IfEmptyTF{#3}{}{\ \IfNoValueTF{#4}{\left}{#4}\vert\ \vphantom{#2}#3\IfNoValueTF{#4}{\right.}{}}}%
    [#4]%
}

\ifboldNumberSets
    \newcommand*{\N}{\mathOrText{\mathbf{N}}}

    \newcommand*{\R}{\mathOrText{\mathbf{R}}}
    
    \newcommand*{\indicatorFunctionSymbol}{\mathbf{1}}
\else
    \newcommand*{\N}{\mathOrText{\mathds{N}}}

    \newcommand*{\R}{\mathOrText{\mathds{R}}}
    
    \newcommand*{\indicatorFunctionSymbol}{\mathds{1}}
\fi

\RenewDocumentCommand{\Pr}{m O{} o}%
{%
    \probabilisticFunctionTemplate{\mathrm{Pr}}{#1}[#2][#3]%
}

\NewDocumentCommand{\E}{m O{} o}%
{%
    \probabilisticFunctionTemplate{\mathrm{E}}{#1}[#2][#3]%
}

\NewDocumentCommand{\Var}{m O{} o}%
{%
    \probabilisticFunctionTemplate{\mathrm{Var}}{#1}[#2][#3]%
}

\DeclareDocumentCommand{\bigO}{m o}%
{%
    \functionTemplate{\mathcal{O}}{(}{)}{#1}[#2]%
}

\DeclareDocumentCommand{\smallO}{m o}%
{%
    \functionTemplate{\mathrm{o}}{(}{)}{#1}[#2]%
}

\DeclareDocumentCommand{\bigTheta}{m o}%
{%
    \functionTemplate{\upTheta}{(}{)}{#1}[#2]%
}

\DeclareDocumentCommand{\bigOmega}{m o}%
{%
    \functionTemplate{\upOmega}{(}{)}{#1}[#2]%
}

\DeclareDocumentCommand{\smallOmega}{m o}%
{%
    \functionTemplate{\upomega}{(}{)}{#1}[#2]%
}

\DeclareDocumentCommand{\eulerE}{o}%
{%
    \mathOrText{\mathrm{e}\IfNoValueTF{#1}{}{^{#1}}}%
}

\DeclareDocumentCommand{\poly}{m o}%
{%
    \functionTemplate{\mathrm{poly}}{(}{)}{#1}[#2]%
}

\createFunction{\id}{\mathrm{id}}

\NewDocumentCommand{\ind}{m o o}%
{%
    \IfNoValueTF{#2}%
    {%
        \mathOrText{\indicatorFunctionSymbol_{#1}}%
    }%
    {%
        \functionTemplate{\indicatorFunctionSymbol_{#1}}{(}{)}{#2}[#3]%
    }%
}

\DeclareDocumentCommand{\dom}{m o}%
{%
    \functionTemplate{\mathrm{dom}}{(}{)}{#1}[#2]%
}

\DeclareDocumentCommand{\rng}{m o}%
{%
    \functionTemplate{\mathrm{rng}}{(}{)}{#1}[#2]%
}

\DeclareDocumentCommand{\d}{o}%
{%
    \mathrm{d}\IfNoValueTF{#1}{}{^{#1}}%
}

\DeclareDocumentCommand{\set}{m m o}
{
    \mathOrText{\IfNoValueTF{#3}{\left}{#3}\{#1\ \IfNoValueTF{#3}{\left}{#3}\vert\
    \vphantom{#1}#2\IfNoValueTF{#3}{\right.}{}\IfNoValueTF{#3}{\right}{#3}\}}
}      

\definecolor{mygreen}{RGB}{1, 150, 122}
\definecolor{myblue}{RGB}{64, 54, 247}

\definecolor{hpired}{RGB}{177,6,58}

\newcommand{\Wlog}[0]{w.l.o.g. }
\newcommand{\NP}{\mathrm{NP}}
\newcommand{\ooea}{\textrm{(1+1) EA}\xspace}

\newcommand{\Bern}[1]{\mathrm{Bern}\left(#1\right)}
\newcommand{\Bin}[2]{\mathrm{Bin}\left(#1,#2\right)}
\newcommand{\Geo}[1]{\mathrm{Geo}\left( #1 \right)}
\newcommand{\Ex}[1]{\E{ #1 }}
\newcommand{\ceil}[1]{\left\lceil #1 \right\rceil}
\newcommand{\floor}[1]{\left\lfloor #1 \right\rfloor}
\renewcommand{\epsilon}{\varepsilon}
\newcommand{\with}[0]{~\vert~}

\newcommand{\stocProc}[1]{\left( #1_t\right)_{t \in \N}}
\renewcommand{\N}{\mathbb{N}}
\renewcommand{\R}{\mathbb{R}}

\newcommand{\stNeigh}[1]{\mathcal{N}\left[ #1 \right]}
\newcommand{\fitness}[0]{f}
\newcommand{\minS}[1]{\min\left\{ #1 \right\}}

\newcommand{\clNeigh}[1]{N\left[ #1 \right]}
\newcommand{\oNeigh}[1]{N\left( #1 \right)}
\newcommand{\pn}[2]{\text{pn}\left[ #1 , #2 \right]}
\newcommand{\Dk}[0]{\mathcal{D}_k\left(C_n\right)}
\newcommand{\D}[0]{\mathcal{D}\left(C_n\right)}
\newcommand{\redVert}[1]{\mathcal{R}\left( #1 \right)}
\newcommand{\adj}[1]{\mathcal{A}\left( #1 \right)}

\newcommand{\cw}[2][]{\textrm{cw}_{#1}\left( #2 \right)}
\newcommand{\ccw}[2][]{\textrm{ccw}_{#1}\left( #2 \right)}
\newcommand{\dccw}[2]{d_{\textrm{ccw}}\left(#1, #2 \right)}
\newcommand{\dcw}[2]{d_{\textrm{cw}}\left(#1, #2 \right)}

\newcommand{\effRes}[2]{\mathcal{R}\left( #1 \leftrightarrow #2 \right)} 

\begin{document}

    \frontmatter

\ifprintVersion
    \ifprofessionalPrint
        \newgeometry
        {
            textwidth = 134 mm,
            textheight = 220 mm,
            top = 38 mm + \extraborderlength,
            inner = 38 mm + \mybindingcorrection + \extraborderlength,
        }
    \else
        \newgeometry
        {
            textwidth = 134 mm,
            textheight = 220 mm,
            top = 38 mm,
            inner = 38 mm + \mybindingcorrection,
        }
    \fi
\else
    \newgeometry
    {
        textwidth = 134 mm,
        textheight = 220 mm,
        top = 38 mm,
        inner = 38 mm,
    }
\fi

\begin{titlepage}
    \sffamily
    \begin{center}
        \includegraphics[height = 3.2 cm]{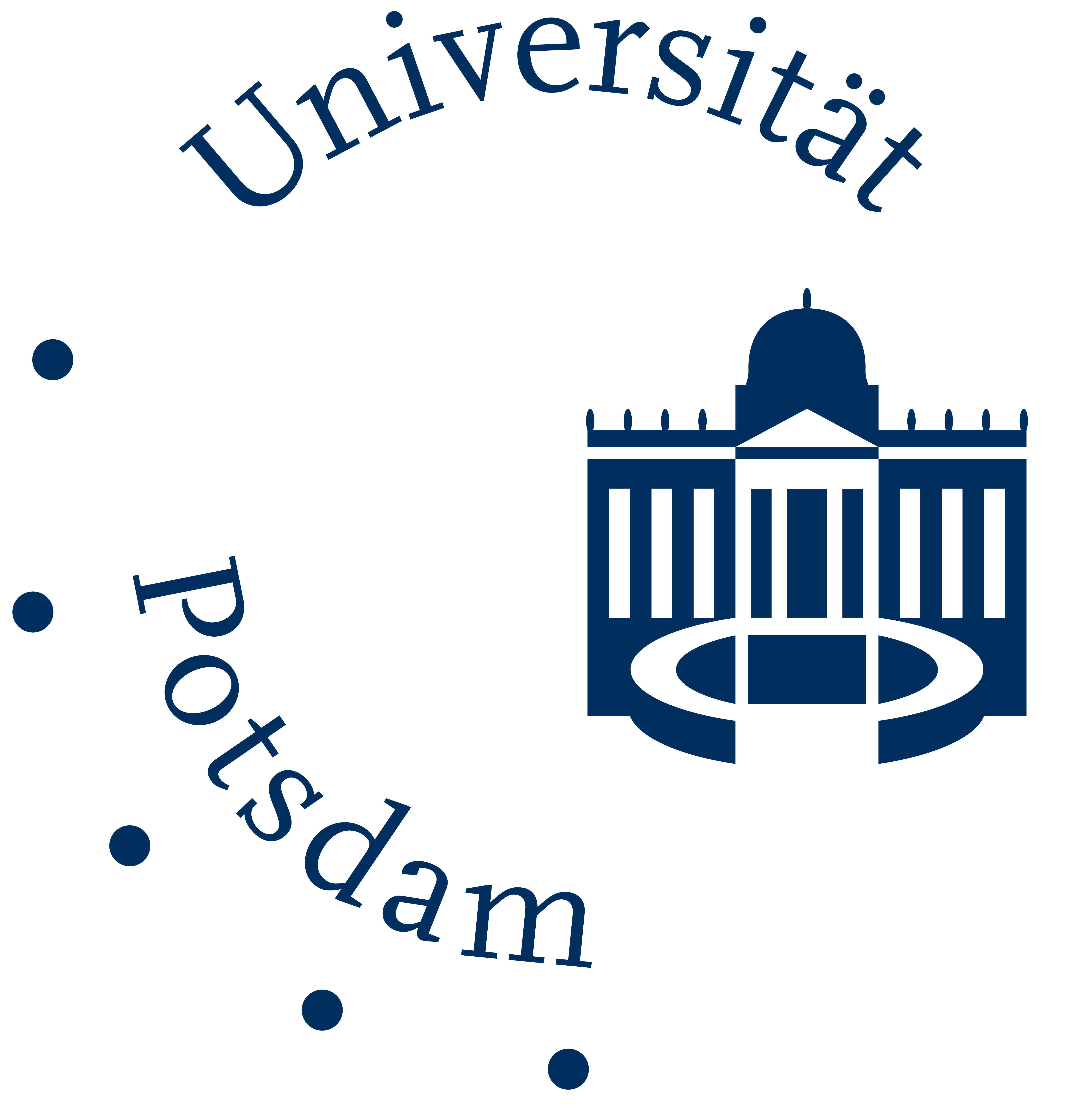} \hfill \includegraphics[height = 3 cm]{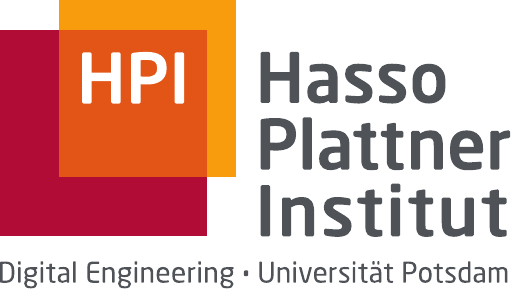}\\
        \vfil
        {\LARGE
            \rule[1 ex]{\textwidth}{1.5 pt}
            \onehalfspacing\printTitleBold\\[1 ex]
            {\vspace*{-1 ex}\Large \printGermanTitle}\\
            \rule[-1 ex]{\textwidth}{1.5 pt}
        }
        \vfil
        {\Large\textbf{\printAuthor}}
        \vfil
        {\large Universitäts\colloquialDegreeNameLowercase arbeit\\[0.25 ex]
        zur Erlangung des akademischen Grades}\\[0.25 ex]
        \bigskip
        {\Large \colloquialDegreeName{} of Science}\\[0.5 ex]
        {\large\emph{(\degreeAbbreviation\,Sc.)}}\\
        \bigskip
        {\large im Studiengang\\[0.25 ex]
        \printProgram}
        \vfil
        {\large eingereicht am \printDateReceived{} am\\[0.25 ex]
        Fachgebiet Algorithm Engineering der\\[0.25 ex]
        Digital-Engineering-Fakultät\\[0.25 ex]
        der Universität Potsdam}
    \end{center}
    
    \vfil
    \begin{table}[h]
        \centering
        \large
        \sffamily 
        {\def\arraystretch{1.2}
            \begin{tabular}{>{\bfseries}p{3.8 cm}p{5.3 cm}}
                Gutachter               & \printNameOfSupervisor\\
                Betreuer                & \printAdditionalExaminers
            \end{tabular}
        }
    \end{table}
\end{titlepage}

\restoregeometry

    \pagestyle{plain}

    \addchap{Abstract}

Dominating Set is a well-known combinatorial optimization problem which finds application in computational biology or mobile communication. Because of its $\NP$-hardness, one often turns to heuristics for good solutions. Many such heuristics have been empirically tested and perform rather well. However, it is not well understood why their results are so good or even what guarantees they can offer regarding their runtime or the quality of their results. For this, a strong theoretical foundation has to be established. We contribute to this by rigorously analyzing a Random Local Search (RLS) algorithm that aims to find a minimum dominating set on a graph. We consider its performance on cycle graphs with $n$ vertices. We prove an upper bound for the expected runtime until an optimum is found of $\bigO{n^4\log^2(n)}$. In doing so, we introduce several models to represent dominating sets on cycles 
that help us understand how RLS explores the search space to find an optimum.
For our proof we use techniques which are already quite popular for the analysis of randomized algorithms. We further apply a special method to analyze a reversible Markov Chain, which arises as a result of our modeling. This method has not yet found wide application in this kind of runtime analysis.

    \selectlanguage{ngerman}
    \addchap{Zusammenfassung}

Dominating Set ist ein bekanntes Problem der kombinatorischen Optimierung, welches für die Bioinformatik oder Mobilkommunikation relevant ist. Aufgrund der $\NP$-Schwere des Problems werden oft heuristische Ansätze verfolgt, um gute Lösungen zu finden. Viele Heuristiken wurden bereits empirisch getestet und liefern oft gute Ergebnisse. Allerdings ist wenig darüber bekannt, warum die Ergebnisse dieser Heuristiken so gut sind, oder welche Garantien es an deren Laufzeit oder die Qualität der Lösungen gibt. Dafür braucht es eine stabile theoretische Grundlage. Wir tragen mit dieser Arbeit dazu bei, indem wir eine randomisierte lokale Suche analysieren, die kleinste dominierende Mengen auf Graphen finden soll. Dafür betrachten wir ihre Laufzeit auf Kreisgraphen mit $n$ Knoten. Wir beweisen, dass die erwartete Laufzeit, bis eine kleinste dominierende Menge gefunden wurde, in $\bigO{n^4\log^2(n)}$ liegt. Dafür führen wir mehrere Modelle ein, um dominierende Mengen auf Kreisgraphen darzustellen und betrachten, wie sich diese in Bezug auf Mutationen der lokalen Suche verhalten. Dadurch erlangen wir ein besseres Verständnis dafür, wie die lokale Suche den Lösungsraum erkundet, um zum Optimum zu finden. Für unseren Beweis verwenden wir Techniken, die für die Analyse von randomisierten Algorithmen bereits weit verbreitet sind. Darüber hinaus wenden wir eine spezielle Methode zur Analyse einer reversiblen Markov-Kette an, die sich aus unserer Modellierung ergibt. Diese Methode hat für diese Art von Laufzeitanalyse noch keine breite Anwendung gefunden.
    \selectlanguage{american}


    \setuptoc{toc}{totoc}
    \tableofcontents

    \pagestyle{headings}
    \mainmatter

    \chapter{Introduction}

\emph{Randomized Algorithms} play a large role when it comes to finding solutions to hard problems. It is widely believed that, for the elements of the class $\NP$, no efficient algorithms exist. 
However, finding solutions to these problems is of great practical relevance: determining where to best open a new store (Facility Location), finding a shortest roundtrip between several cities (Traveling Salesperson), or where to build radio towers for optimal coverage (Dominating Set) are only a handful of examples. Since these problems need to be solved, regardless of their complexity, one has to take different approaches. Some of these include limiting the domain to instances where these problems are tractable \cite{blasius_solving_2023,hedetniemi_linear_1986} or accepting solutions which are not exact, i.e. approximate solutions which are some distance away from the optimum \cite{hochba_approximation_1997}. 

To utilize both approaches, we turn to \emph{Heuristics}. They offer guiding information on how a good solution can be found in less time than exact methods \cite{eiselt_heuristic_2000}. To do so, they can combine approaches that work well on certain problem instances with ones that offer robust guarantees on the quality of the resulting solution. A simple greedy heuristic is to always add an element to the solution that maximizes a certain objective function. Examples of such heuristics include the $2$-approximation for Vertex Cover and the $1 + \log(n)$-approximation for Dominating set \cite{ausiello_complexity_1999,chvatal_greedy_1979}. Another heuristic approach is to not always look for the \emph{best} next candidate solution, but just for one that is \emph{better}.

A prominent algorithm that operates according to this paradigm is \emph{Random Local Search} (RLS). Like other randomized search heuristics, it has found application in real problems such as clustering and inventory routing \cite{franti_randomised_2000,benoist_randomized_2011}, where it provides good results in a short amount of time while still being easy to implement. However, most research regarding such randomized heuristics is strictly empirical, leaving, as \citeauthor{droste_analysis_2002} put it, even the best known results subject to controversial discussion \cite{droste_analysis_2002}. They stress the importance of a rigorous theoretical basis to mitigate this. 

As its name implies, the RLS heuristic traverses the solution space of a problem by exploring only those solutions that are in some sense close to its current candidate \cite{aarts_local_2003}.
Contrary to other randomized search algorithms, such as the \ooea, the neighborhood of a solution is not the whole solution space. This makes the analysis of the trajectory of such algorithms even more dependent on their current solutions and the corresponding neighborhood, as now some transitions are not only improbable but outright impossible. We will use the well-known concept of Markov Chains to model these local dependencies.

In \Cref{chap:2} we provide the necessary concepts in graph and probability theory that we use throughout this thesis. 
We define the concept of networks in \Cref{chap:3} which relate to reversible Markov Chains. We then analyse a random walk on a triangular grid which arises when analysing the trajectory of the RLS in the subsequent chapter. We begin \Cref{chap:4} by discussing important properties of Dominating Sets on cycle graphs. We leverage these to prove a bound for the RLS discovering a dominating set that already has only $1.5$ times the size of the optimum in expected time $\bigO{n \log(n)}$ in \Cref{thm:half}.  
From there, we develop several models to better understand the trajectory of the RLS when trying to optimize further. We capture when it is possible to swap some vertices and how single swaps help the RLS move towards dominating sets than can further be reduced. 
Using these models, we prove a bound of $\bigO{n^4\log^2(n)}$ for the expected time RLS improves this solution to an optimal one in \Cref{thm:optimal}. 

\begin{section}{Dominating Set}
    The problem Dominating Set asks, given a simple undirected graph, for a subset of vertices such that every vertex is either in this subset or has a neighbor in it. The optimization variant asks for such a set with smallest cardinality over all sets that have this property. This problem finds application in computational biology, to identify proteins that play a major role in human diseases in large networks of protein-protein interactions \cite{nacher_minimum_2016}. In a similar fashion, in ad-hoc wireless network routing, where connections can be unreliable, finding a dominating set of the network can be used to determine those nodes that are most likely to successfully deliver a package to its recipient \cite{samuel_dtn_2009}. In document summarization, where a graph can represent sentences and their similarity, a minimum dominating set may be interpreted as a maximally representative and minimally redundant summary of the contents \cite{shen_multi-document_2010}.
    
    The problem is, however, well-known to be $\NP$-complete and as such believed not to be efficiently solvable. It is even $\NP$-complete on graph classes that are easy for other $\NP$-hard problems, such as bipartite, chordal or split graphs, as surveyed by \citeauthor{haynes_fundamentals_2023} \cite{haynes_fundamentals_2023}. 
    It is even hard to approximate within a factor of $(1-\varepsilon)\ln(n)$ for any $\varepsilon > 0$ \cite{chlebik_approximation_2008}. A $1+\log(n)$-approximation is achieved by a simple greedy algorithm that, in every iteration, adds that vertex to the set that covers the most uncovered nodes \cite{chvatal_greedy_1979}.

    Other approximation heuristics have been developed and tested for this problem: \citeauthor{sanchis_experimental_2002}~\cite{sanchis_experimental_2002} analyses five heuristics on known graph examples as well as random graphs. All of them are essentially variants of the well-known greedy heuristic. On her test set of random graphs, the initial greedy heuristics still offered the best results, most of the time actually finding an optimal dominating set. On another class of random graphs with a predetermined domination number, another procedure incorporating some problem knowledge performed slightly better.

    These heuristics aimed only to find good, not optimal, dominating sets. Another approach, more in the realm of randomized heuristics, is presented by \citeauthor{hedar_hybrid_2010}. They combine a classic genetic framework, i.e., randomly altering solutions and combining them to obtain new ones, with some knowledge specific to the problem of Dominating Set. Their heuristics aim to continually improve their solution to eventually arrive at an optimum. In each iteration, a Local Search procedure removes a random vertex from the dominating set and adds one not in the set if the result is no longer a dominating set. A filtering procedure successively removes each vertex from the dominating set and computes the resulting fitness value, removing possibly redundant vertices. A selection mechanism they call \say{Elite Inspiration} finds the intersection of the best dominating sets and extends this solution until it becomes dominating again, to identify the most relevant vertices. They test this algorithm on randomly generated graphs and conclude that it significantly outperforms a genetic algorithm that only randomly mutates and combines solutions each iteration \cite{hedar_hybrid_2010}. 

    Other heuristics that have been applied to this problem include Ant-colony Optimization \cite{ho_enhanced_2006}, Evolutionary Algorithms with different mutation operators \cite{chaurasia_hybrid_2015} or other Local Search procedures \cite{chalupa_order-based_2018, casado_iterated_2023}. All of these papers provide an empirical analysis of their algorithms.
    To the best of our knowledge, no rigorous theoretical analysis of a random search heuristic has been conducted for the problem of Dominating Set.  To introduce this topic and some related techniques, we present other instances where this has been done.
\end{section}   

\begin{section}{Analysis of Randomized Search Heuristics}
    Randomized search heuristics, like Random Local Search, are often analysed on pseudo-boolean functions that have certain characteristics \cite{droste_analysis_2002}. This helps researchers understand what properties help these heuristics in discovering good solutions and what might hinder them. However, we concern ourselves with their application to a problem rooted in graph theory. Other prominent problems with such roots, for which there are theoretical time complexity results, include Maximum Matchings \cite{giel_evolutionary_2003}, Minimum Spanning Trees \cite{neumann_randomized_2007}, Graph Coloring \cite{sudholt_analysis_2010}, and Vertex Cover \cite{baguley_analysis_2022}. 
    
P    Their analysis often involves some \say{simple} graph classes. \citeauthor{sudholt_analysis_2010} prove results of a local search operator for coloring on bipartite, planar, and sparse random graphs with a low chromatic number. \citeauthor{baguley_analysis_2022} analyse their operator for Vertex Cover on paths and bipartite graphs. These analyses are, however, not motivated by a need to discover ever faster algorithms, but again by an interest to better understand the behaviour of randomized search heuristics.

    As such, several techniques have been developed and utilized for these analyses. For evolutionary algorithms such as the \ooea, stochastic drift can often aid in their analysis. Drift and respective results formalize the idea that if you know how fast a process moves towards a target, you can infer how long it takes to reach the target. \citeauthor{baguley_analysis_2022} use it to show how long a heuristic takes to sample a proper vertex cover of a graph \cite{baguley_analysis_2022}. \citeauthor{doerr_drift_2011} \cite{doerr_drift_2011} also shows how drift arguments can yield a simpler proof with additional bounds for the Minimum Spanning Tree problem tackled by \citeauthor{neumann_randomized_2007} \cite{neumann_randomized_2007}.

    However, this simpler proof still utilizes many results regarding the general problem structure. Analyses of such heuristics often incorporate specific domain knowledge.
    In \cite{giel_evolutionary_2003}, a crucial part of the runtime proof also relies on known results about matchings in graphs and how they can be improved. \citeauthor{giel_evolutionary_2003} then use these to argue how the \ooea can leverage these with specific random mutations. \citeauthor{neumann_randomized_2007} use properties of non-minimum spanning trees to show how much their search heuristic can improve the solution \cite{neumann_randomized_2007}. This suggests that the underlying problem structure can be used to aid the analysis of these heuristics. We will employ many of these techniques in this thesis to analyse the performance of our Random Local Search algorithm. We now introduce the necessary preliminaries.

\end{section}



    \chapter{Preliminaries}\label{chap:2}

In this chapter, we introduce the definitions and concepts in graph theory and probability theory that are necessary for this thesis. We also define and motivate our Random Local Search heuristic. We first state some basic definitions. Let $\N \coloneq \{0,1,2,\dots\}$ be the set of natural numbers including $0$. We write $\N^+$ for the set $\N \setminus \{0\}$. By $[n]$ we denote the set $\{1,\dots,n\}$. For pairs of natural numbers $(a,b) \in \N^2$ we define $|(a,b)| = a + b$. For a finite set $S$ we denote its cardinality by $|S|$. Finally, by $\log(\cdot)$ we denote the natural logarithm and by $\log^2(x)$ we mean $\log(x)\log(x)$.

\begin{section}{Graph Theory and Dominating Set}
    A graph $G = (V,E)$ is a pair of a vertex set $V$ and an edge set $E \subseteq V \times V$. For vertices $u,v \in V$ with $\{u,v\} \in E$ we say that they are adjacent or neighbors and denote this by $x \sim y$. We allow vertices $v \in V$ to have self loops and denote them by $v \sim v$.
    
    For a vertex $v \in V$, by $\oNeigh{v}$ we denote its open neighborhood $\oNeigh{v} \coloneq \{u \with \{u,v\} \in E\}$. The closed neighborhood of a vertex $v \in V$ is the set $\clNeigh{v} \coloneq N(v) \cup \{v\}$. We define the closed neighborhood of a set $T \subseteq V$ as $\clNeigh{T} \coloneq \bigcup_{v \in T} \clNeigh{v}$.
    
    We consider two simple graph classes, paths and cycles.
    \begin{definition}[Path]
        A path of length $k \in \N^+$ is a graph with vertex set $V \coloneq \{v_1,\dots,v_k\}$ and edge set $E \coloneq \bigcup_{1\leq i < k} \{v_i, v_{i+1}\}$. We denote it by $P_k$.
    \end{definition}

    \begin{definition}[Cycle]
        A cycle of length $k \in \N^+$ is a path $P_k$ with an additional edge $\{v_k, v_1\}$. We denote it by $C_k$.
    \end{definition}

    A dominating set of a graph $G = (V,E)$ is a subset of its vertices $D \subseteq V$ such that for every vertex $v \in V$, either $v \in D$ or there exists a vertex $u \in D$ with $\{u,v\} \in E$. An equivalent formulation is that $\clNeigh{D} = V$. We call a dominating set a minimum dominating set when there is no other dominating set with strictly smaller cardinality. We call a dominating set $D$ a minimal dominating set when there is no proper subset $S \subset D$, such that $S$ is also a dominating set.
    \newpage
    We note the following result for the size of a minimum dominating set on cycles, as given in \cite{haynes_fundamentals_2023} \emph{(Observation 2.21)}: 
    \begin{proposition}
        A minimum dominating set of $C_n$ for $n \in \N^+$ has cardinality $\ceil{\frac{n}{3}}$.
    \end{proposition}
    Intuitively, one can select every third vertex to cover all vertices. Then every vertex covers its own neighborhood uniquely if $n$ is a multiple of $3$, and this still holds for most vertices if $n$ has a remainder $1$ or $2$ modulo $3$.
\end{section}

\begin{section}{Random Local Search Algorithm}
    Random Local Search differs from other randomized search heuristics in that it only explores solutions that are somewhat \say{close} to it. To define what close means, we first have to define how we will represent our solutions.

    For a graph $G = (V,E)$ with $|V| = n$, we represent a solution or state of the search algorithm as a bitstring $y \in \{0,1\}^n$. Given some numbering of the vertices, the $i$-th bit defines whether the $i$-th vertex is part of the solution. For a solution $y \in \{0,1\}^n$, we define $S(y)$ to be the set induced by this solution, namely $S(y) \coloneq \{v_i \in V \with y_i = 1\}$. By $|y|_1$ we denote the number of bits set to $1$ in state $y$ and by $|y|_0$ the number of zero bits, respectively. 

    A typical measure of distance for two bitstrings is the Hamming distance, i.e., the number of bits where they differ. For RLS, one might consider exploring only those states with Hamming distance $1$ to the current state. In every step, the algorithm would then select a random bit and flip it.

    To decide whether to accept a new solution or to continue exploring from the old one, the algorithm has to judge how good a solution is. For this, we have to provide a so-called fitness function. This term stems from the roots of randomized search heuristics in natural evolution. We choose a simple function that will prioritize finding a feasible solution, i.e., an actual dominating set, and after that help guide the search towards a minimum dominating set. We define \[
    \fitness\colon \{0,1\}^n \rightarrow \N, x \mapsto |\{v \in V \with \clNeigh{v} \cap S(x) = \emptyset\}| \cdot (n+1) + |S(x)|
    \] as our fitness function. The first term weighs the number of uncovered vertices, those that are neither in the set nor have a neighbor in the set, with a factor of $n+1$. That way, every non-dominating set will have fitness at least $n+1$, which is greater than the maximum cardinality of a dominating set.
    
    These two considerations would yield an algorithm that could successfully solve the Dominating Set problem on some graph classes. However, as \citeauthor{neumann_randomized_2007}~\cite{neumann_randomized_2007} put it, only considering the neighborhood with Hamming distance $1$ often makes heuristics susceptible to many local optima which are not globally optimal. As is the case with our current algorithm. Consider the path $P_3$ with both endpoints selected. Then this is a valid dominating set, which is not the optimal solution consisting of only the middle vertex. However, no single flip will be accepted, since removing one of the selected vertices leaves this vertex uncovered, and adding the middle vertex increases the cardinality of the set. We therefore include some states with Hamming distance $2$. We allow the algorithm to choose a vertex in the current set and then select one from its neighborhood, removing one and adding the other. We limit selection to the immediate neighborhood, as this will simplify our analysis. With this, we define our Random Local Search in \Cref{alg:rls}.

\SetKwIF{WithProb}{ElseWithProb}{WithProbElse}{with probability}{do}{else with probability}{else}{endWithProbability}

\begin{algorithm}
    \caption{Random Local Search Algorithm}\label{alg:rls}
    Choose $x \in \{0,1\}^n$ uniformly at random\;
    \While{not happy}
    {
    \uWithProb{$1/2$}
    {
        pick a vertex $v \in V(G)$ u.a.r. \;
        $y \gets$ flip $x_v$\;
    }
    \WithProbElse
    {
        choose $u \in S$ and $v \in \oNeigh{v}$ u.a.r. \;
        \If{$\overline{x_v}$}
        {
            $y \gets$ flip $x_u$ and $x_v$\;
        }
    }
    \If{$\fitness(y) \leq \fitness(y)$} {
        $x \gets y$\;
    }
    }
\end{algorithm}

We refer to the part in the code where a vertex is chosen at random and flipped as the flip operator, or say that the RLS chose to flip a vertex. Similarly, the second part is the swap operator.

We note that \emph{not happy} is a placeholder for an actual stopping criterion. In practice, one could stop the algorithm after a specific time. However, we look at this algorithm from a purely theoretical point of view and are only interested in the expected time to find an optimal solution, so we let the search run indefinitely. 

\end{section}

\begin{section}{Probability Theory}

    We expect a rudimentary understanding of probability theory for most of this thesis. We refer the reader to \cite{mitzenmacher_probability_2017} for basic definitions. 

    We will often consider so-called stochastic processes, where we look at a collection of random variables over time.

    \begin{definition}[Stochastic Process]
        A stochastic process is a collection of random variables over a common sample space $\Omega$, indexed by times $t \in \N$. We usually denote it by $\stocProc{X}$. If $\Omega = \R$, we call this process integrable if for $t \in \N, \Ex{X_t}$ is finite.
    \end{definition}

    We use the following drift theorem, as given by \citeauthor{kotzing_theory_2024}~\cite{kotzing_theory_2024}, where it can be found as \emph{Theorem 2.5}.

    \begin{theorem}[Multiplicative Drift] \label{thm:mult_drift}
  Let $(X_t)_{t \in \N}$ be an integrable process over $\{0,1\} \cup S$ where $S \subset \R_{>1}$, and let  $T = \min\{t \in \N \mid X_t \leq 0\}$. Assume that there is a $\delta \in \R_+$ such that, for all $s \in S \cup \{1\}$ and all $t < T$ it holds that
  \[
  \Ex{X_{t} - X_{t+1} ~|~ X_0,\dots,X_t} \geq \delta X_t.
  \]
  Then
  \[
  \Ex{T} \leq \frac{1 + \log{\Ex{X_0}}}{\delta}.
  \]
\end{theorem}

In the following sections, we also consider special stochastic processes, namely Markov Chains. Those are processes obeying a certain property (adapted from \cite{mitzenmacher_probability_2017}, \emph{Definition 7.1}):

\begin{definition}[Markov Chain]
    A stochastic process $\stocProc{X}$ is a Markov Chain if 
    \[
    \Pr{X_{t+1} = x_{t+1}}[X_0=x_0,\dots,X_t=x_t] = \Pr{X_{t+1}=x_{t+1}}[X_t=x_t].
    \]
    For $u,v \in \Lambda$, where $\Lambda$ is the sample space of the Markov Chain, we refer to the transition probabilities as \[
    P(u,v) = \Pr{X_{t+1} = v}[X_t = u].
    \]
    We call the sample space of the Markov Chain its state space.
\end{definition}

We later make use of the equation of Wald, as stated in \cite{mitzenmacher_probability_2017} \emph{(Theorem 13.3)}.

\begin{theorem}[Wald's Equation]\label{thm:wald}
    Let $X_1,X_2,\dots$ be nonnegative, independent, and identically distributed random variables with distribution $X$. Let $T$ be a stopping time for this sequence. If $T$ and $X$ have bounded expectation, then \[
    \Ex{\sum_{i=1}^T X_i} = \Ex{T}\Ex{X}.
    \]
    Here, $T$ is a stopping time if the event $T = n$ is independent of $X_{n+1},X_{n+2},\dots$.
\end{theorem}
    
\end{section}

    \chapter{Markov Chains}\label{chap:3}
    In this chapter, we look at a model for Markov Chains with a special property. We restate important definitions and known results and further give some intuition on how this approach works. We especially focus on how it can help analyze some random walks. We then provide an example that will later become relevant during the runtime analysis for the RLS. We prove new results for this specific instance.

\begin{section}{Markov Chains as Networks}\label{sec:markov_networks}

 We adapt all of our definitions from \citeauthor{levin_markov_2017}~\cite{levin_markov_2017}. We will indicate where the respective definitions, theorems and propositions can be found in their book. The model that we will introduce here are so called networks. We first formally define them:

\begin{definition}[Network]\label{def:network}
    \emph{(Introduction to Section 9.1)}~
    A network is a finite undirected connected graph $G$ with vertex set $V$ and edge set $E$, and a map from edges to positive real numbers $c\colon E\rightarrow \R^+$, where $c(e)$ is called the conductance of an edge $e$. For an edge $\{x,y\} \in E(G)$ we often write $c(x,y)$ for $c(\{x,y\});$ clearly $c(x,y) = c(y,x)$. The reciprocal $r(e) = 1 / c(e)$ is called the resistance of the edge $e$. A network will be denoted by the pair $(G, c)$. We often call the vertices of $G$ \emph{nodes}. We call a network \emph{even} when the conductances of all edges, except for self loops, are equal. We call this conductance the uniform conductance of the network. Equivalently, we define the uniform resistance of the network.
\end{definition}

The naming of the associated values to the edges is already quite suggestive. And indeed, this model actually considers the underlying graph as an electrical network, where every edge is a resistor with a fixed resistance. This allows for the use of many results from physics regarding the behaviour of currents in such a network. This is useful to us, because there are certain probabilistic interpretations of these results, many of which relate to random walks on this network:

\begin{definition}[Weighted random walk on a network]\label{def:random_walk_network}
 \emph{(Introduction to Section 9.1)}~
    Consider the Markov Chain on the vertices of $G$ with transition matrix \[
    P(x,y) = \frac{c(x,y)}{c(x)},
    \]
    where $c(x) = \sum_{y: y \sim x} c(x,y)$. This process is called the weighted random walk on $G$ with edge conductances $c$.
\end{definition}

In our study of networks, we will explicitly allow them to have self loops. When a vertex $v$ has a loop to itself, we write $v \sim v$. In the random walk, this corresponds to a self-loop, where the random walk may rest at a vertex. We can use these networks to model random walks with a nice property that is referred to as reversibility.

\begin{definition}[Reversibility]\label{def:reversibility}
    We call a Markov Chain $\stocProc{X}$ on a state space $\Lambda$ reversible with respect to a probability distribution $\pi$ on $\Lambda$, if for every $x,y\in \Lambda$ \[
        \pi(x)\cdot\Pr{X_{t+1}=y}[X_t = x] = \Pr{X_{t+1}=x}[X_t = y]\cdot \pi(y).
    \]
\end{definition}

An important result also highlighted by \citeauthor{levin_markov_2017} \emph{(End of Section 9.1)}~ is that every reversible Markov Chain can be modeled by a random walk on a network. The intuition behind that is the following: Let $\stocProc{X}$ be a reversible Markov Chain with respect to $\pi$. Let $\Lambda$ be its state space and $P'(\cdot)$ its transition probabilities. Further assume that for all states $x,y \in \Lambda$, $\Pr{\exists t \in \N\colon X_t = y}[X_0 = x] > 0$, so that every state can transition to any other state in finitely many steps. This ensures the following graph will be connected. We can define a network $G$ with the underlying state space as its vertex set and an edge whenever the transition probability from one to the other is positive. For nodes $x,y \in \Lambda$ such that $\{x,y\} \in E(G)$ as $c(x,y)= \pi(x)P'(x,y)$ and similarly for self loops as $c(x,x) = \pi(x)P'(x,x)$. By definition, we then have $c(x) = \sum_{y:x \sim y} c(x,y) = \pi(x)$, since the transition probabilities in the product of the conductances sum to $1$. Then the transition probability for nodes $x,y$ is exactly \[P(x,y) = \frac{c(x,y)}{c(x)} = \frac{\pi(x)P'(x,y)}{\pi(x)} = P'(x,y)\] which is exactly the transition probability of the original Markov Chain. Since we already mentioned currents, we have to state two concepts relating to those, which will be relevant for other theorems.

\begin{definition}[Flow]\label{def:flow}
    \emph{(Introduction of Section 9.3)}~
    Let $(G,c)$ be a network.
    A flow $\theta$ is a function on oriented edges which is antisymmetric, meaning that for an edge $\{x,y\} \in E(G)$: $\theta(\vec{xy}) = -\theta(\vec{yx})$, where $\vec{xy}$ is the directed edge from $x$ to $y$ and $\vec{yx}$ is defined analogously. 
    We define the divergence of $\theta$ at $x$ by \[
    \textrm{div}\theta(x) \coloneq \sum_{y:y \sim x} \theta(\vec{xy}).
    \]
    Let $s,t$ be distinguished vertices called \emph{source} and \emph{sink}, respectively. A \emph{flow from $s$ to $t$} must satisfy
    \begin{enumerate}
        \item $\textrm{div}\theta(x) = 0$ at all $x \not\in \{s,t\}$, and
        \item $\textrm{div}\theta(s) \geq 0$.
    \end{enumerate}
    Note that 1. is the classic \say{flow in equals flow out} requirement.
    We define the \emph{strength} of a flow $\theta$ from $s$ to $t$ as $\textrm{div}\theta(s)$. A \emph{unit flow} from $s$ to $t$ is a flow from $s$ to $t$ with strength $1$.
\end{definition}

\begin{definition}[Energy]\label{def:energy}
    \emph{(Preliminary of Theorem 9.10)}~
    Let $(G,c)$ be a network and let $s,t \in V(G)$.
    For a flow $\theta$ from $s$ to $t$, we define the energy of the flow as \[
    \mathcal{E}(\theta) \coloneq \sum_e [\theta(e)]^2 r(e).
    \]
\end{definition}

With the concept of energy, we can define a metric between nodes in the network, the so called effective resistance. We only define it as a result of the following minimization, although it has a probabilistic interpretation as well. The effective resistance between to nodes $a$ and $b$ is proportional to the probability of a random walk on the network (see \Cref{def:random_walk_network}) started from $a$ visiting $b$ before returning to $a$.

\begin{theorem}[Thomson's principle]\label{thm:thomson}
    \emph{(Theorem 9.10)}~
    Let $(G,c)$ be a network.
    For source and sink vertices $s,t \in V(G)$ in any finite, connected graph \[
    \mathcal{R}(s \leftrightarrow t) = \inf\{\mathcal{E}(\theta) \with \theta \text{ is a unit flow from} s \text{ to } t\}.
    \]
    There exists a unique minimizing flow from $s$ to $t$ for the $\inf$ above.
    We call the quantity $\mathcal{R}(s \leftrightarrow t)$ the \say{effective resistance} of $s$ and $t$.
\end{theorem}

The effective resistance also has a useful connection to expected hitting times, given by the following theorem.

\begin{definition}[Commute Time]\label{def:commute_time}
    \emph{(Introduction of Section 10.3)}~
    The commute time between nodes $a,b$ in a network is the expected time to move from $a$ to $b$ and then back to $a$. We denote by $\tau_{a,b}$ the (random) amount of time to transit from $a$ to $b$ and then back to $a$. That is, for a random walk $\stocProc{X}$ on the network,
    \[
    \tau_{a,b} = \minS{t \geq \tau_b : X_t = a},
    \]
    where $X_0 = a$. The commute time is then \[
    t_{a \leftrightarrow b} \coloneq \Ex{\tau_{a,b}}[X_0 = a].
    \]
\end{definition}

\begin{theorem}[Commute time identity]\label{thm:commute_time_identity}
    \emph{(Proposition 10.7)}~
    Let $(G,c)$ be a network, and let $\stocProc{X}$ be the random walk on this network. For any nodes a and b in $V(G)$,
    \[
    t_{a\leftrightarrow b} = c_G\effRes{a}{b},
    \]
    where $c(x)=\sum_{y:y\sim x} c(x,y)$ and $c_G = \sum_{x \in V(G)} c(x)$. 
\end{theorem}

The commute time is also, by definition, an upper bound on the expected time to reach node $b$ starting from $b$.

The following theorem will allow us to reason about networks that have other networks, for which there are known results, as subgraphs. Intuitively, you can model a network, which is a subgraph of another one, with the same edges and \say{infinite} resistance. This will greatly reduce the complexity of our analysis.

\begin{theorem}[Rayleigh's monotonicity law]\label{thm:rayleigh}
    \emph{(Theorem 9.12)}~
    If $r$ and $r'$ are two assignments of resistances to the edges of the same graph $G$ that satisfy $r(e) \leq r'(e)$ for all $e \in E(G)$, then \[
    \mathcal{R}(s \leftrightarrow t;r) \leq \mathcal{R}(s \leftrightarrow t; r').
    \]
\end{theorem}

In our definition of networks, we introduced the new concept of an even network, where all resistances except those of self-loops are the same. We prove that for bounds on the effective resistance, we can always consider a unit network, where all edges have resistance $1$ and factor in the resistance into the bounds on this alternate network. This will simplify some argumentation and make the calculations less cluttered. 

\begin{lemma}\label{lem:even_unit}
    Let $(G, c)$ be an even network with uniform resistance $R$. Let $c'$ be another even assignment of conductances to the same edges with uniform conductance/resistance $1$. Then for any $s,t$
    \[
    \mathcal{R}(s \leftrightarrow t; c) = R \cdot \mathcal{R}(s \leftrightarrow t; c').
    \]
\end{lemma}
\begin{proof}
    We want to show that the unique minimizer of energy in the network with conductances $c$ is also the minimizer of energy with conductances $c'$. We prove this by contradiction. 
    Let $\hat{\theta}$ be the unique flow that minimizes the energy w.r.t. $c$. Suppose that the minimizing flow of the energy w.r.t. $c'$ was a different flow $\hat\theta'$ and thus let $\mathcal{E}(\hat\theta';c') < \mathcal{E}(\hat\theta;c')$. 
    
    Then \begin{align}
         \mathcal{E}(\hat\theta'; c)&= \sum_e [\hat{\theta}'(e)]^2 r(e) = \sum_e [\hat\theta'(e)]^2 R
        =R\cdot\sum_e [\hat{\theta}'(e)]^2 \label{eq:energy}\\
        \notag&= R\cdot\mathcal{E}(\hat\theta';c')  \\
        \intertext{which then, by our assumption, is}
        \notag&< R\cdot\mathcal{E}(\hat\theta; c')
        = R\cdot\sum_e{[\hat\theta(e)]}^2 = \sum_e [\hat\theta(e)]^2
r(e) = \mathcal{E}(\hat\theta;c).    \end{align}
which is a contradiction, since $\hat\theta$ was a minimizer of energy w.r.t. $c$. So $\mathcal{E}(\hat\theta; c') \leq \mathcal{E}(\hat\theta'; c')$. Since there is no flow with strictly less energy and the energy has a unique minimizing flow, $\hat\theta$ is also the unique minimizer of energy w.r.t. $c'$. From our calculations in \eqref{eq:energy} and the definition of effective resistance according to \Cref{thm:thomson} it then follows immediately that \[
\mathcal{R}(s \leftrightarrow t; c) = R\cdot\mathcal{R}(s \leftrightarrow t; c').
\]
\end{proof}

From now on, whenever we consider even networks, we simply analyze the network with unit conductances and resistances. The following Lemma is often omitted from some sources because of its simplicity. We state it to gain a better intuition and give an example for a simple bound on the effective resistance.

\begin{lemma}\label{lem:resistance_distance}
    In an even network $(G,c)$ with uniform resistance $R$, for  distinct vertices $u,v \in V(G)$ \[
    \effRes{u}{v} \leq R \cdot d(u,v).
    \]
    where $d(\cdot)$ is the hop distance in $G$.
\end{lemma}
\begin{proof}
    We will consider the even network with the same underlying graph and unit conductances. \Cref{thm:thomson} tells us that the effective resistance is bounded above by any unit flow from $u$ to $v$. We will therefore construct such a flow $\theta$ to prove our bound. Consider a shortest path $P$ from $u$ to $v$ in the sense that it traverses the fewest edges. Then this path has $d(u,v)$ edges. Orient these edges away from $u$ in the direction of the path and let $\theta(e) = 1$ and $\theta(e^-) = -1$ for each directed edge on the path, where $e^-$ is the flipped orientation of $e$. All other edges are assigned $0$. $\theta$ is a proper flow from $u$ to $v$ and has unit strength. Further, \[
    \mathcal{E}(\theta) = \sum_e [\theta(e)]^2 r(e) = \sum_e[\theta(e)]^2 = \sum_e \mathbb{1}[e \text{ lies on } P] = d(u,v).\
    \]
    Here, $\mathbb{1}$ denotes the indicator function that is $1$ if its argument is an edge of $P$ and $0$ otherwise. Thereby and with \Cref{lem:even_unit}, $\effRes{u}{v} \leq R\cdot d(u,v)$.
\end{proof}

The last nice property of the effective resistance is that it is actually a metric on the nodes of a network and satisfies a triangle inequality. We can thus construct bounds for the resistance from known bounds for some subgraphs or subpaths.

\begin{theorem}[Triangle inequality]\label{thm:resistance_triangle_ineq}
    \emph{(Corollary 10.8)}~
    The effective resistance $\mathcal{R}$ satisfies a triangle inequality: If $a,b,c$ are vertices, then \[
    \effRes{a}{c} \leq \effRes{a}{b} + \effRes{b}{c}.
    \]
\end{theorem}

\end{section}

\begin{section}{Resistance in Triangle Grids}

We will now use the concepts introduced in the previous section to analyze a very specific example. We want to bound the effective resistances in an even network on a Triangle Grid graph with uniform resistance $R$. We first define what we mean by that. After that, we derive an upper bound in $\bigO{R \cdot \log^2(n)}$ for the resistance by looking at the resistances in subgraphs. We choose these subgraphs by constructing special paths from one vertex to another, which
allow us to apply a known result about the resistance in Square Grid graphs. 

\begin{definition}[Triangle Grid Graph]\label{def:triangle_grid}
    For $n \in \N$, the Triangle Grid graph $T_n$ is the graph with vertex set $V \coloneq \{(x,y) \in \N^2 \with x + y \leq n\}$ where two vertices are adjacent if they are a distance of $1$ apart, or have a distance of $2$ and lie on the same diagonal. We call the vertex corresponding to $(0,0)$ the origin of this graph. For $v \in V$, we define $\partial(v,k) \coloneq \{u \in V \with d(v,u) \geq k\}$.
    We define $\chi_n \coloneq |V(T_n)|$.
\end{definition}

See \Cref{fig:triangle_grid} for an illustrated Triangle Grid graph. We note the following result regarding the cardinality of the vertex set of this Triangle Grid graph. This will become relevant in \Cref{chap:4}.
\begin{lemma}\label{clm:cardinality_triangle}
   For $n \in \N^+$, $\chi_n = \frac{(n+1)(n+2)}{2}$. 
\end{lemma}
\begin{proof}
    We consider the vertex set of $T_n$. Then there is one vertex with $y$-coordinate $n$, two vertices with $y$-coordinate $n-1$, and so forth until finally, there are $n + 1$ vertices with $y$ coordinate $0$. By definition, no vertices have strictly smaller or larger $y$-coordinates than $0$ and $n$. Thus, we have \[
    |V(T_n)| = \sum_{y=0}^{n} n+1-y = \sum_{i = 1}^{n+1} i = \frac{(n+1)(n+2)}{2}.
    \]
\end{proof}

\begin{figure}
    \centering
    \includegraphics[width=0.4\linewidth]{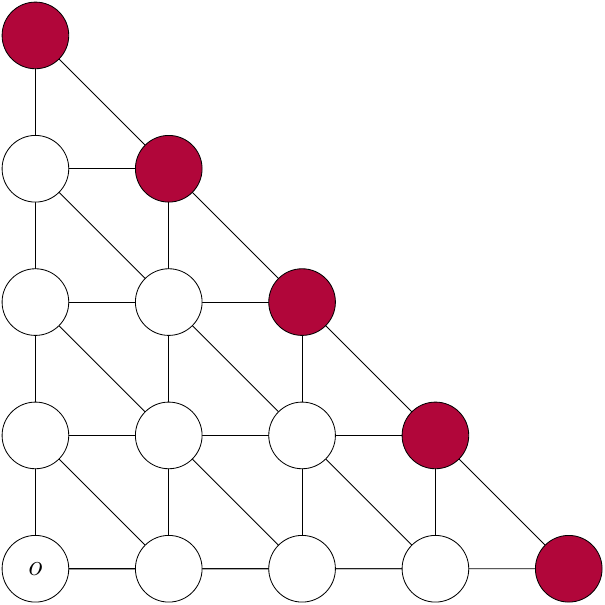}
    \caption{The Triangle Grid graph $T_4$. The origin is marked $o$. The vertices in $\partial(o,4)$ are highlighted.}
    \label{fig:triangle_grid}
\end{figure}

For the proof of the bound, we will need the following result on the effective resistance between cornerpoints of a Square Grid:

\begin{definition}[Square Grid Graph]\label{def:square_grid}
    \emph{(Introduction of Section 9.5)}~
    Let $B_n$ be the $n \times n$ two-dimensional grid graph with $V \coloneq \{(z,w) \in \N^2 \with 1 \leq z,w \leq n\}$, where vertices are adjacent when they are at unit Euclidean distance.
\end{definition}

\begin{theorem}[Resistances in squares]\label{thm:resistance_square}
    \emph{(Proposition 9.17)}~
    Let $a = (1,1)$ be the lower left-hand corner of $B_n$, and let $z = (n,n)$ be the upper right-hand corner of $B_n$. Suppose each edge of $B_n$ has unit conductance. The effective resistance $\effRes{a}{z}$ satisfies \[
    \frac{\log(n)}{2} \leq \effRes{a}{z} \leq 2\log(n).
    \]
\end{theorem}

We can now prove the following lemma.

\begin{lemma}\label{lem:resistance_triangle_origin}
    Let, for $n \in \N^+$, $(T_n, c)$ be an even network with uniform resistance $R$. Let $o$ be the origin of $T_n$ and $t \in V(T_n) \setminus \{o\}$. Then \[
    \effRes{o}{t} \leq R \cdot \big(8 \log^2(n) + 6(\log(n) + 1)\big)
    \]
\end{lemma}
\begin{proof}
    We will again consider the network with $T_n$ as the underlying graph endowed with unit conductances for all edges. 
    We prove this statement by proving the following invariant w.r.t. the number of rows of the triangle grid: Let $k \in \N$. Then, for some $r,s \in \N\colon~ k = 4s + r$. Let $o$ be the origin of $T_k$ and $t \in \partial(o,k)$. Let $r_k \coloneq \effRes{o}{t}$ in the network with $T_k$ as the underlying graph. Then \begin{align}
        r_k = r_{4s + r} \leq 4\log(s + 1)\log_2(4s) + 3(\log_2(4s+r) + 1). \label{eq31}
    \end{align}
    We first prove this for $1 \leq k \leq 3$. By \Cref{lem:resistance_distance}, $r_k \leq k \leq 3$, which is accounted for by the second term of the sum. Thus, the inequality holds. 
    \par Now let $k$ be arbitrary but fixed, such that for all $1 \leq k' < k$ \eqref{eq31} holds. 
    Let $r,s \in \N$, with $0 \leq r < 4$ such that $k = 4s + r$. Let $t = (t_x, t_y)$ and let $t_y \leq 2s$. Since $k = t_x + t_y$, this has $t_x \geq 2s + r$.
    We now consider points $a \coloneq (s,s), b \coloneq (2s, 0), c \coloneq (2s + r,0)$. Since the sums of their coordinates are all less than $4s + r$, they are all in $V(T_k)$.
    By \Cref{thm:resistance_triangle_ineq}, we have \[
        r_k \leq \effRes{o}{a} + \effRes{a}{b} + \effRes{b}{c} + \effRes{c}{t}.
    \]

    \begin{figure}
        \centering
        \includegraphics[width=0.4\linewidth]{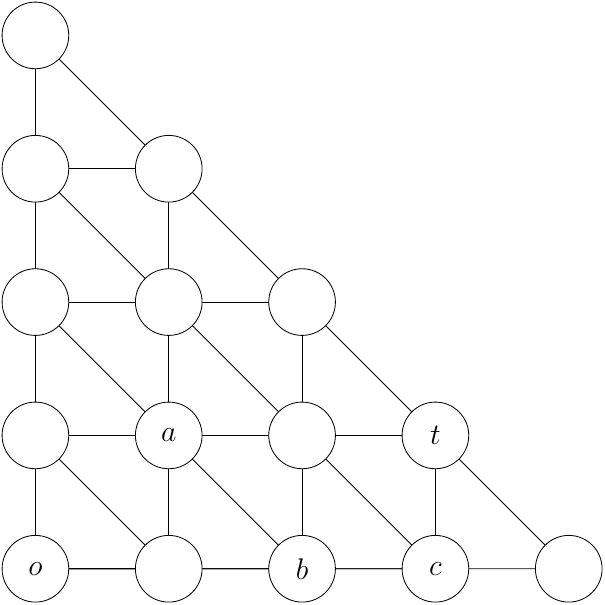}
        \caption{The Triangle Grid graph $T_4$ with its origin $o$. To determine the effective resistance $\effRes{o}{t}$, \Cref{lem:resistance_triangle_origin} uses vertices $a,b,c$ and known bounds for the induced square grids.}
        \label{fig:triangle_paths}
    \end{figure}
    
    For a concrete example of this construction, see \Cref{fig:triangle_paths}.
    We first consider points $o$ and $a$. They constitute corner-points of the subgraph $B_{s+1}$ in $T_k$. Since in $T_k$ only edges along the diagonals are added, by \Cref{thm:rayleigh} we know that $\mathcal{R}(o \leftrightarrow a; T_k) \leq \mathcal{R}(o \leftrightarrow a; B_{s+1})$. This means that we can apply \Cref{thm:resistance_square} to bound the effective resistance. The same holds for vertices $a$ and $b$.
    To bound the effective resistance between $b$ and $c$ we again use \Cref{lem:resistance_distance}. Now we only need to concern ourselves with the effective resistance between $c$ and $t$. For this, we consider the subgraph $T_{2s}$ of $T_k$ with $c$ as its origin. We note that the energy of the energy minimizing unit flow from $c$ to $t$ is necessarily less than or equal to the energy of an energy minimizing unit flow from $c$ to $t$ that only uses edges in $T_{2s}$, since the infimum is simply taken over a restricted class of flows. We therefore have \begin{align*}
    r_k = r_{4s+r} &\leq 2\log(s + 1) + 2\log(s + 1) + d(b,c) + r_{2s} \\
    &\leq 4\log(s + 1)+  r + r_{2s} \\
    &\leq 4\log(s + 1) + 3 + r_{2s}. \\
    \intertext{Since $2s < k$, we can now apply our invariant \eqref{eq31}:}
    &\leq 4\log(s + 1) + 3 + 4\log\Big(\floor{\frac{2s}{4}} + 1\Big)\log_2(2s) + 3(\log_2(2s) + 1) \\
    &\leq 4\log(s + 1) + 4\log(s + 1)\log_2(2s)+ 3(\log_2(2s) + 1) + 3 \\
    &= 4\log(s + 1)(\log_2(2s) + 1) + 3(\log_2(2s) + 1 + 1) \\
    &= 4\log(s + 1)(\log_2(2s) + \log_2(2)) + 3(\log_2(2s) + \log_2(2) + 1) \\
    &= 4 \log(s + 1)\log_2(4s) + 3(\log_2(4s) + 1) \\
    &\leq 4\log(s + 1)\log_2(4s) + 3(\log_2(4s + r) + 1).
    \end{align*}
    When $t_y \geq 2s + r$, this has $t_x \leq 2s$. We can now make the same argument with points $a \coloneq (s,s), b\coloneq(0, 2s), c \coloneq (0, 2s + r)$, which admits the same bound. 
    \par We finally consider the case that $2s < t_y < 2s + r$.  We say $t_y = 2s + x$ for $0 < x < r$. Then, let $a \coloneq (s, s), b \coloneq (2s, 2s)$. With \Cref{lem:resistance_distance}, vertex $b$ has \begin{align*}
    \effRes{b}{t} \leq d(b,t) &\leq |t_x - b_x| + |t_y - b_y| = |(4s + r - t_y) - 2s| + |t_y - 2s| \\
    &= |2s + r - t_y| + |t_y - 2s| \\
    &= |2s + r - (2s + x)| + |2s + x - 2s| \\
    &= |r - x| + |x| = r.
    \end{align*}
    Since $o$ and $a$ and $a$ and $b$ are corners of a $B_{s+1}$ subgraph, with \Cref{thm:resistance_square}, they have $\effRes{o}{a} \leq 2\log(s + 1)$ and this holds likewise for $\effRes{a}{b}$. With \Cref{thm:resistance_triangle_ineq}, we finally have \begin{align*}
    \effRes{o}{t} &\leq \effRes{o}{a} + \effRes{a}{b} + \effRes{b}{t} \\
    &\leq 4\log(s + 1) + 3 \\
    &\leq 4\log(s+1)\log_2(4s) + 3(\log_2(4s +r) +1).
    \end{align*}
    
    Thus, our invariant also holds for $k$.
    To conclude, note that $\log_2(x) = \frac{\log(x)}{\log(2)} \leq \frac{\log(x)}{\frac{1}{2}} \leq 2\log(x)$ for $x \in \R^+$ and for $n \in \N^+$ with $n = 4s + r$ and  $0 \leq r < 4~\colon~ s + 1 \leq n$. Finally, with \Cref{lem:even_unit}, we get that \[
    \effRes{o}{t} = R \cdot r_n \leq R \cdot \Big(8\log^2(n) + 6(\log(n) + 1)\Big).
    \]
\end{proof}
\end{section}

    \chapter{RLS for Dominating Set on Cycle Graphs}\label{chap:4}
    \newcommand{\logb}[1]{\log_2\left(#1\right)}

In this chapter, we analyze the performance of our RLS algorithm for solving the Minimum Dominating Set problem on cycle graphs. We utilize techniques we have developed in \Cref{chap:3} and other tools from probability theory. We start off by considering the time it takes the RLS to first sample a feasible dominating set, no matter its quality. After we enter the region of feasibility, we spend the next section discovering some nice properties of large dominating sets and how they help the RLS to quickly improve the quality of the solution. In the following section, we introduce a new way to represent dominating sets on the cycle: The adjacency model. We analyze how the changes of the solution by the RLS are reflected in our new model. We then formalize another representation in the next section: the Particle Systems. They abstract dependencies of the previous models, allowing us to better reason about their behaviour. In the final section, we use these models to analyze the runtime of the RLS to find an optimal solution, proving a bound of $\bigO{n^4\log^2(n)}$.

\begin{section}{Time to Feasibility}

Recall that our RLS algorithm (\Cref{alg:rls}) initializes a set at random, i.e., it includes every vertex independently with probability $\frac{1}{2}$. Our first result proves that this random initialization will most likely not result in a dominating set.

\begin{lemma}\label{thm:initial_non_dominating}
    A randomly initialized set $y \in \{0,1\}^n$  does not induce a dominating set on the cycle graph $C_n$ with probability at least $1 - 2^{-\Theta(n)}$.
\end{lemma}

\begin{proof}
    We consider the case where $n \mod 3 = 0$ first and discuss the other cases later.
    We can partition the graph for $0 \leq i < \frac{n}{3}$ into $3$-tuples $(v_{3i}, v_{3i + 1}, v_{3i+2})$. Let $M_i$ be the event that the vertex $v_{3i+1}$ is not dominated. For this, all three vertices of the $i$-th tuple must not be included in our solution. We therefore have $
    \Pr{M_i} = \frac{1}{2}\cdot\frac{1}{2}\cdot\frac{1}{2} = \frac{1}{8},
    $ since every vertex is selected independently with probability $\frac{1}{2}$.
    The probability that this middle vertex is dominated is therefore $
    \Pr{\overline{M_i}} = 1 - \Pr{M_i} = \frac{7}{8}.
    $
    \par We now argue that these events are mutually independent, i.e., it holds that for every subset $I \subseteq [n]$ with $|I| \geq 2\colon
    \Pr{\bigcap_{i \in I} \overline{M_i}} = \prod_{i \in I}\Pr{\overline{M_i}}.
    $
    Let $I \subseteq [n]$ be any such subset. The intersection of all events $(\overline{M_i})_{i \in I}$ asks that all vertices $(v_{3i+1})_{i \in I}$ are dominated.
    A single vertex $v$ can only dominate those vertices in its closed neighborhood, i.e., those vertices $u$ with $d(u,v) \leq 1$. 
    For $i,j \in I$, a vertex $v$ that dominates $v_{3i+1}$ can therefore not dominate $v_{3j+1}$ as $d(v_{3j+1}, v_{3i+1}) \geq 3$ and $d(v, v_{3i+1}) \leq 1$ meaning $d(v,v_{3j+1}) \geq 2$. 
    Thus, whether vertices $v_{3i},v_{3i+1},v_{3i+2}$ are selected only influences the domination of $v_{3i+1}$. 
    As the selection of a vertex of our solution is independent of all other vertices, event $\overline{M_i}$ is independent of all other events $\overline{M_j}$. 
    For $j \in I$ and event $\overline{M_j}$, we can therefore argue that
    $\Pr{\bigcap_{i \in I} \overline{M_i}} = \Pr{\overline{M_j}} \cdot \Pr{\bigcap_{i \in I \setminus \{j\}}\overline{M_j}}$. 
    By applying this reasoning iteratively, we see that $\Pr{\bigcap_{i \in I} \overline{M_i}} = \prod_{i\in I}\Pr{\overline{M_i}}$
    
    Let $D$ be the event that a randomly initialized solution is a dominating set for the graph. It is necessary that such a proper solution dominates all middle vertices with respect to the partitioning described above. Thus the probability thereof is an upper bound for the probability of event $D$.   
    We now have
    \[\Pr{D} \leq \Pr{\bigcap_{0\leq i < \frac{n}{3}}\overline{M_i}} = \prod_{0\leq i < \frac{n}{3}} \overline{M_i} = \left(\frac{7}{8}\right)^{\frac{n}{3}} = 2^{\frac{\logb{7} - \logb{8}}{3}n} = 2^{-\Theta(n)}\]
    It follows that
    \[
    \Pr{\overline{D}} = 1 - \Pr{D} \geq 1 - 2^{-\Theta(n)}
    \]
    For the case $n \mod 3 \neq 0$, we ignore the vertices that cannot be part of a $3$-tuple. As a consequence, $\Pr{\bigcap_{0\leq i <\left\lfloor\frac{n}{3}\right\rfloor}\overline{M_i}}$ remains an upper bound for $\Pr{D}$ since this is still a necessary condition for the solution to induce a dominating set. Since $\frac{n}{3} - \left\lfloor \frac{n}{3}\right\rfloor < 1$, we still have a bound within $\Theta(n)$ and the statement holds. 
\end{proof}

With this we can already prove our first bound regarding the runtime, namely the expected time until the RLS samples a feasible dominating set.

\begin{theorem}\label{thm:feasible}
    Starting from a randomly initialized set $y \in \{0,1\}^n$, the RLS takes in expectation $\bigO{n \log (n)}$ fitness evaluations to sample a feasible set.
\end{theorem}
\begin{proof}
    Consider $y \in \{0,1\}^n$ initialized uniformly at random. As per \Cref{thm:initial_non_dominating}, $S(y)$ is not a dominating set with probability at least $1 - 2^{-\Theta(n)}$. Thus we consider only this case for the expected runtime. Let $N(y)$ be the set of non-dominated vertices. 
    Let $x_t$ denote the current solution of the RLS just before iteration $t+1$; we let $X_t = |N(x_t)|$. For this process to decrease, the number of non-dominated vertices has to be reduced. For one such vertex to become dominated, two events can occur: Either the vertex itself or one of its direct neighbors is selected into the set. For a vertex $v \in N(y)$, let $D_v$ denote the event, that this vertex becomes dominated. We consider only the possibility that the RLS selects one vertex uniformly at random to flip it and that this vertex is $v$. This happens with probabilities $\frac{1}{2}$ and $\frac{1}{n}$ respectively. Then $\Pr{D_v} \geq \frac{1}{2n}$. Considering this and ignoring the possibility that a selected vertex dominates more vertices than itself, we can bound the expected progress of the process by
    \[
    \Ex{X_{t+1} - X_t \mid X_0,\dots,X_t} \geq 1\cdot \Pr{X_{t+1}- X_t = 1\mid X_0, \dots,X_t}\} \geq \frac{1}{2n} \cdot X_t.
    \]
    Let $I$ be the random variable denoting the number of initially selected vertices. Then $I \sim \Bin{n}{\frac{1}{2}}$ and $\Ex{I} = \frac{n}{2}$. Since at least all selected vertices are dominated, we have $\Ex{X_0} \leq \Ex{n - I} = n - \Ex{I} = \frac{n}{2}$.
    Let $T = \min\{t \in \N \mid X_t \leq 0\}$ be the time the process first reaches $0$. Then with \Cref{thm:mult_drift} we get \[
    \Ex{T} \leq \frac{1 + \ln \left( \frac{n}{2}\right)}{\frac{1}{2n}} = 2n\cdot \left(1+\ln\left(\frac{n}{2}\right)\right) = 2n + 2n\ln\left(\frac{n}{2}\right) \in \bigO{n \log(n)}
    \]
\end{proof}

Once we have sampled a feasible dominating set, our fitness function $\fitness$ ensures that we never leave the space of feasible solutions again. From now on, every mutation has to result in a valid dominating set as well. The only way to produce a solution with strictly better fitness is to reduce the cardinality of this set. This can, by definition of the algorithm, only be accomplished by flipping a vertex that is not \say{needed} by the set. In the following section, we formalize this concept, discover some interesting properties of dominating sets on cycle graphs, and use them to prove the next runtime bound.

\end{section}
\begin{section}{Important Properties and Early Optimization}

The following two definitions can be found in \citeauthor{haynes_fundamentals_2023}~\cite{haynes_fundamentals_2023} as \emph{Definition 2.10} and \emph{2.13} respectively, and will be important in determining whether certain vertices are relevant for a dominating set.
    
\begin{definition}[Private Neighborhood]
    For a set $S \subseteq V$ and a vertex $v \in S$, the $S$-private Neighborhood of $v$ is the set $\clNeigh{v} \setminus \clNeigh{S \setminus \{v\}}$ and is denoted by $\pn{v}{S}$. For any vertex $v \in S$, we call another vertex $w \in \pn{v}{S}$ an $S$-private neighbor of $v$.
\end{definition}

\begin{definition}[Redundant Set]
    A set $S \subseteq V$ is called redundant, when there exists a $v \in S$ for which $\pn{v}{S} = \emptyset$. We call $S$ irredundant, if it is not redundant. We define $\redVert{S} \coloneq \{v \in S \with \pn{v}{S} = \emptyset\}$. We call $v \in \redVert{S}$ a redundant vertex.
\end{definition}

For a visual example of private neighbors and redundant sets in graphs see \Cref{fig:redundant_sets}. We now prove a simple property of redundancy.

\begin{lemma}\label{lem:redundancy_superhereditary}
    For any set $S \subset S' \subseteq V$ and $u \in S$, if $\pn{u}{S} = \emptyset$, then $\pn{u}{S'} = \emptyset$.
\end{lemma}
\begin{proof}
    Let $S \subset S' \subseteq V$ and $u \in S$ with $\pn{u}{S} = \emptyset$. Then, for all $v \in \clNeigh{u}$, there exists another vertex $w \in S$ different from $u$, such that $v \in \clNeigh{w}$. These vertices are also elements of $S'$, since $S'$ is a proper superset of $S$. Thus $\pn{u}{S'} = \emptyset$.
\end{proof}

\begin{figure}
    \centering
    \begin{minipage}{.4\textwidth}
        \centering
        \begin{tikzpicture}[line width=0.5pt, node distance=1.3cm]
        \node [circle, draw, fill=hpired] (A) {\color{white}{$A$}};
        \node [circle, draw, fill=hpired] (B) [below of=A,xshift=0.5cm] {\color{white}{$B$}};
        \node [circle, draw] (D) [right of=B] {$D$};
        \node [circle, draw, fill=hpired] (C) [above of=D,xshift=0.5cm] {\color{white}{$C$}};
        \node [circle, draw] (E) [right of=C,yshift=0.7cm] {$E$};
        \path (A) edge (C);
        \path (B) edge (D);
        \path (A) edge (B);
        \path (C) edge (D);
        \path (C) edge (E);
    \end{tikzpicture}
    \end{minipage}
    \begin{minipage}{.4\textwidth}
        \centering
        \begin{tikzpicture}[line width=0.5pt, node distance=1.3cm]
        \node [circle, draw] (A) {$A$};
        \node [circle, draw, fill=hpired] (B) [below of=A,xshift=0.5cm] {\color{white}{$B$}};
        \node [circle, draw] (D) [right of=B] {$D$};
        \node [circle, draw, fill=hpired] (C) [above of=D,xshift=0.5cm] {\color{white}{$C$}};
        \node [circle, draw] (E) [right of=C,yshift=0.7cm] {$E$};
        \path (A) edge (C);
        \path (B) edge (D);
        \path (A) edge (B);
        \path (C) edge (D);
        \path (C) edge (E);
        \end{tikzpicture}
    \end{minipage}
    \caption{Left: The set $S \coloneq \{A,B,C\}$. $E$ is an $S$-private neighbor of $C$. $S$ is redundant because $A$ is a redundant vertex. Right: The set $S' \coloneq \{B,C\}$. $B$ is now an $S$-private neighbor of itself. Both $B$ and $C$ have private neighbors. $S'$ is thus irredundant.}
    \label{fig:redundant_sets}
\end{figure}
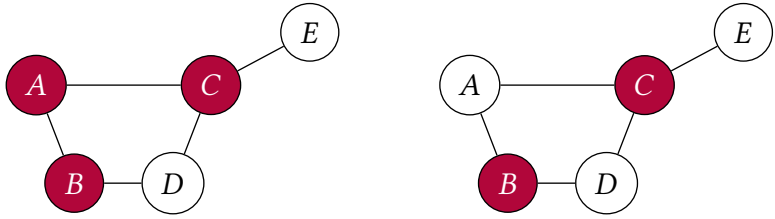

A lot of the time, we will need to reason about dominating sets of a fixed cardinality. We now introduce some shorthand notation for these sets.

\begin{definition}
    Let $n,k \in \N^+$. We denote the set of all dominating $k$-sets of $C_n$ by $\Dk$. We denote all dominating sets of $C_n$ by $\D$.
\end{definition}

\begin{lemma}\label{lem:minimal_redundant}
    For $n,k \in \N^+$, a dominating Set $D \in \Dk$ is minimal iff it is irredundant.
\end{lemma}
\begin{proof}
    Let $D \in \Dk$ be a minimal dominating set. Then, for any $S \subset D$, $S$ is not dominating. Then especially, for some $v \in D$, $D \setminus \{v\}$ is not dominating. That means that there must be a vertex $w \in V$ such that it is dominated by $v$ but not $D \setminus \{v\}$. Therefore $w \in \pn{v}{D}$. Since no vertex in $D$ has a non-empty private neighborhood, $D$ is irredundant. 
    \par Conversely, let $D \in \Dk$ be an irredundant dominating set. Then every vertex $v \in D$ has a non-empty private neighborhood. This means that \[
    \clNeigh{D \setminus \{v\}} = \clNeigh{D} \setminus \pn{v}{D} \subset V,
    \] 
    so $D \setminus \{v\}$ is not a dominating set. Thus every subset of $D$ is not a dominating set and $D$ is minimal.
\end{proof}

This lemma can also be found in \cite{cockayne_properties_1978} as \emph{Proposition 4.1}. However, as \citeauthor{cockayne_properties_1978} omit the proof, we chose to state it here. For a dominating set on a cycle to be redundant, we need to fully cover the neighborhood of a vertex with other vertices. We introduce the following definition to specify what we mean by this. We use this to prove a necessary and sufficient condition for a dominating set to be redundant.

\begin{definition}[Arc]\label{def:arc}
    Let $n \in \N^+$. For $k \leq n$, we call an induced path $P_k$ in $C_n$ an arc of length $k$.
\end{definition}

\begin{lemma}\label{lem:distance_redundant}
    For $n \in \N^+$, a dominating set $D \in \Dk$ is redundant iff there exists an arc $A$ within the $C_n$ with length at most 4 and $|A \cap D| \geq 3$. 
\end{lemma}
\begin{proof}
    Let $D \in \Dk$ be a redundant dominating set. Then there exists a vertex $v \in D$ with $\pn{v}{D} = \emptyset$. This has $\clNeigh{D} = \clNeigh{D \setminus \{v\}}$. Since $\clNeigh{v} = \{u,v,w\}$ for $u,w \in V$ adjacent to $v$, there have to be vertices $x,y \in D$ different from $v$, such that $\{u,v,w\} \subseteq \clNeigh{x} \cup \clNeigh{y}$. Since then $v \in \clNeigh{x} \lor v \in \clNeigh{y}$, $v$ has to be a neighbor of either $x$ or $y$. Let \Wlog $x$ be that vertex and $x = u$. Then $w \not\in \clNeigh{x}$ and therefore $w \in \clNeigh{y}$. Since $w$ has to be in the neighborhood of $y$ and $d(v,w) = 1$, we have $d(v,y) \leq 2$. Then $x,v,y$ are contained in an arc of length at most 4.
    \par Conversely, let $D \in \Dk$ be a dominating set with a set of vertices $\{u,v,w\} \subset D$ which are contained in an arc $A$ of length at most 4. Let $v$ lie in between of $u$ and $w$. Since $d(u,w) \leq d(u,v) + d(v,w) < |A| \leq 4$ and both distances are greater than $0$, one of these distances has to be equal to $1$. Let \Wlog $d(u,v) = 1$ and $d(v,w) \leq 2$. Then $\clNeigh{v} \subset \clNeigh{u} \cup \clNeigh{w}$. This has $\pn{v}{D} = \emptyset$ which means $D$ is redundant.
\end{proof}

This suggests that if we only have enough vertices in the dominating set, it will surely be redundant. At some point the vertices have to be so packed together that we will certainly meet this constraint. The following lemma proves this intuition. 

\begin{lemma}\label{lem:big_ds_not_minimal}
    For $n \in \N^+, k > \floor{\frac{n}{2}}$, every $D \in \Dk$ is not minimal.
\end{lemma}
\begin{proof}
    Let $k > \floor{\frac{n}{2}}$ and $S \subset V$ with $|S| = k$. Consider the set of all clockwise, consecutive $4$-tuples of vertices, i.e., $T \coloneq \{(v_i, v_{i+1}, v_{i+2}, v_{i+3}) \with i \in [n] \}$ where addition is applied modulo $n$. There are $n$ such tuples, so $|T| = n$. 
    Now we consider the set $S' \coloneq S \times [4]$. Thus $|S'| = 4|S|$. 
    We define a function $f\colon S' \rightarrow T$, where $f(v,i) = t \in T$ where $t_i = v$. 
    This is a well-defined function, since there is exactly one such tuple for every vertex and position pair. 
    Furthermore, for fixed $v \in S$ and $i,j \in [4]$ with $i \neq j$, $f(v,i) \neq f(v,j)$, by construction.
    We consider the two cases, $n$ being even or odd.
    When $n$ is even, then $\floor{\frac{n}{2}} = \frac{n}{2}$. 
    Then $k \geq \frac{n}{2} + 1$. 
    Thus $|S'| = 4|S| = 4k \geq 4\left(\frac{n}{2}+1\right) = 2n + 4 > 2n$. 
    By the pigeonhole principle, there has to be a $t \in T$ with $|f^{-1}(t)| > 2 $. 
    For $n$ odd, we have $\floor{\frac{n}{2}} = \frac{n-1}{2}$ and $k \geq \frac{n-1}{2}+1$ and thus $|S'| \geq 4\left(\frac{n-1}{2} +1\right) = 2n + 2 > 2n$. 
    Again, there has to exist $t \in T$ with $|f^{-1}(t)| > 2$. 
    We then have $\{(u,i), (v,j),(w,k)\} \subseteq f^{-1}(t)$. 
    By observation, $u,v,w$ are three distinct vertices.
    Since these vertices are all part of the same $4$-tuple, they are contained in an arc of length 4.
    If $S$ is a dominating $k$-set, then by \Cref{lem:distance_redundant} and \Cref{lem:minimal_redundant}, $S$ is not a minimal dominating set.
\end{proof}

At this point we could already use this fact to analyze the expected time the RLS would take to sample a dominating set of cardinality $\floor{\frac{n}{2}}$. \Cref{lem:big_ds_not_minimal} shows that in every dominating set of greater cardinality, there is at least one redundant vertex that the RLS could flip. However, using the simple property about redundant sets we proved in \Cref{lem:redundancy_superhereditary}, we can achieve an even better bound. 

\begin{lemma}\label{lem:more_reducible}
    Let $n \in \N^+$. If, for $q \in \N^+$, $k = \floor{\frac{n}{2}} + q$ and $D \in \Dk$, then $|\redVert{D}| \geq q$.
\end{lemma}
\begin{proof}
    Let $q \in \N^+$ with $\floor{\frac{n}{2}} + q \leq n$ and $D_q \in \Dk$. By \Cref{lem:big_ds_not_minimal} we know that $D_q$ is redundant, i.e., there exists a vertex $v_q \in D_q$ such that $\clNeigh{D_q \setminus \{v_q\}} = V$. Let $D_{q-1} = D_q \setminus \{v_q\}$. Again, recursively by \Cref{lem:big_ds_not_minimal}, there exists a vertex $v_{q-1}$ such that $\clNeigh{D_{q-1} \setminus \{v_{q-1}\}} = V$. Repeating this argument, we stop when we reach $D_0$. We now have distinct vertices $v_q,\dots, v_1$, such that for $1 \leq i \leq q$, $\pn{v_i}{D_i} = \emptyset$. Since $D_1 \subset \dots \subset D_q$, by \Cref{lem:redundancy_superhereditary} we have $\pn{v_i}{D_q} = \emptyset$. This means that all these vertices are redundant in $D_q$ and thus $\{v_1,\dots,v_q\} \subseteq \redVert{D_q}$. This has $|\redVert{D_q}| \geq q$.
\end{proof}

\begin{theorem}\label{thm:half}
    Starting from a dominating set $y \in \{0,1\}^n$ with $|y|_1 > \floor{\frac{n}{2}}$, the RLS takes in expectation $\bigO{n \log(n)}$ fitness evaluations to sample a dominating set $y' \in \{0,1\}^n$ with $|y'|_1 = \floor{\frac{n}{2}}$.
\end{theorem}
\begin{proof}
    Let $y_t$ denote the current solution of the RLS before iteration $t+1$. Let $D_t \subseteq V$ denote the subset of vertices induced by $y_t$. Let $S_t = |D_t| - \floor{\frac{n}{2}}$ and let $T \coloneq \minS{t \in \N \with S_t = 0}$. Note that $|D_T| = \floor{\frac{n}{2}}$. We now consider the drift of this process. By definition of the fitness function $\fitness$, we have $0 \leq |D_t| - |D_{t+1}| \leq 1$ and therefore $0 \leq S_t - S_{t+1} \leq 1$. We have $|D_{t}| > |D_{t+1}|$, and correspondingly $S_t > S_{t+1}$, when the RLS flips a vertex $v \in V$ such that $D_t \setminus \{v\}$ is still a dominating set, i.e., $v \in \redVert{D_t}$. Let $G$ denote the event, that the RLS selects the first operator and chooses a vertex in $\redVert{D_t}$ in iteration $t$. By \Cref{lem:more_reducible}, we have \[
    \Ex{S_t - S_{t+1} \with S_0,\dots,S_t} = \Pr{G} = \frac{1}{2}\cdot \frac{|\redVert{D_t}|}{n} \geq \frac{1}{2}\cdot\frac{|D_t|- \floor{\frac{n}{2}}}{n} = \frac{S_t}{2n}.
    \]
    We note that we can ignore any swaps in this analysis, as \Cref{lem:more_reducible} holds for all dominating sets and swaps cannot reduce the cardinality.  By \Cref{thm:mult_drift} and since $ S_0 = |D_0|- \floor{\frac{n}{2}} \leq n - \floor{\frac{n}{2}}$, we have \[
    \Ex{T} \leq 2n\cdot\Bigg(1+\log\Big(n -\floor{\frac{n}{2}}\Big)\Bigg) \leq 2n + 2n\cdot \log(n) \in \bigO{n \log(n)}
    \]
\end{proof}
\end{section}

A minimum dominating set on a cycle graph $C_n$ has cardinality $\ceil{\frac{n}{3}}$. When sampling a set with $\floor{\frac{n}{2}}$ vertices, the RLS only has to remove one sixth of the remaining ones. However the dominating sets with smaller cardinality loose the nice property of always being redundant. 

\begin{section}{Adjacency Models}

We now analyze how the RLS behaves when its current solution is a dominating set $D$ with cardinality $|D| \leq \floor{\frac{n}{2}}$. As we noted, there exist some dominating sets which are irredundant and also not optimal (see \Cref{fig:irredundant_ds_c6} for an illustrated example). To better reason about the behavior of the RLS, we first introduce some new definitions to formalize another representation of dominating sets on the cycle graph. 

\begin{figure}
    \centering
    \begin{minipage}[t]{.45\textwidth}
        \centering
        \begin{tikzpicture}[line width=0.5pt, node distance=1.3cm]
        \def\r{1.5}
        \def\off{90}
        \node [circle, draw,fill=hpired] (A) at ({\r*cos(0+\off)},{ \r*sin(0+\off)}) {\color{white}{$A$}};
        \node [circle, draw] (B) at ({\r*cos(60+\off)},{ \r*sin(60+\off)}) {$B$};
        \node [circle, draw,fill=hpired] (C) at ({\r*cos(120+\off)},{ \r*sin(120+\off)}) {\color{white}{$C$}};
        \node [circle, draw] (D) at ({\r*cos(180+\off)},{ \r*sin(180+\off)}) {$D$};
        \node [circle, draw,fill=hpired] (E) at ({\r*cos(240+\off)},{\r*sin(240+\off)}) {\color{white}{$E$}};
        \node [circle, draw] (F) at ({\r*cos(300+\off)},{ \r*sin(300+\off)}) {$F$};
        \path (A) edge (B);
        \path (B) edge (C);
        \path (D) edge (C);
        \path (D) edge (E);
        \path (F) edge (E);
        \path (F) edge (A);
        \end{tikzpicture}
        \caption{An irredundant and non-optimal dominating set of $C_6$.}
        \label{fig:irredundant_ds_c6}
    \end{minipage}
    \begin{minipage}[t]{.45\textwidth}
        \centering
        \begin{tikzpicture}[line width=0.5pt, node distance=1.3cm]
        \def\r{1.5}
        \def\off{90}
        \node [circle, draw] (A) at ({\r*cos(0+\off)},{ \r*sin(0+\off)}) {$A$};
        \node [circle, draw,] (B) at ({\r*cos(60+\off)},{ \r*sin(60+\off)}) {$B$};
        \node [circle, draw,fill=hpired] (C) at ({\r*cos(120+\off)},{ \r*sin(120+\off)}) {\color{white}{$C$}};
        \node [circle, draw] (D) at ({\r*cos(180+\off)},{ \r*sin(180+\off)}) {$D$};
        \node [circle, draw,fill=hpired] (E) at ({\r*cos(240+\off)},{\r*sin(240+\off)}) {\color{white}{$E$}};
        \node [circle, draw,fill=hpired] (F) at ({\r*cos(300+\off)},{ \r*sin(300+\off)}) {\color{white}{$F$}};
        \path (A) edge (B);
        \path (B) edge (C);
        \path (D) edge (C);
        \path (D) edge (E);
        \path (F) edge (E);
        \path (F) edge (A);
        \end{tikzpicture}
        \caption{A dominating set of $C_6$. $C$ is clockwise, but not counterclockwise movable.}
        \label{fig:movability_cycle}
    \end{minipage}
\end{figure}

\begin{definition}[\lbrack Counter-\rbrack Clockwise neighbors]\label{def:neighbors}
    For any $n,k \in \N^+$, let $S \subseteq V(C_n)$. Let $v \in S$. Then we call a vertex $u \in V(C_n)$ the clockwise neighbor of $v$ with respect to $S$ if $u \in S$ and $u$ is the closest such vertex when traversing $C_n$ in clockwise order, starting from $v$. We say $u = \cw[S]{v}$. Analogously, we define the counterclockwise neighbor $\ccw[S]{v}$ of $v$. We write $\cw{v}$ and $\ccw{v}$ if $S = V(C_n)$ or the reference set is apparent from context.
\end{definition}

\begin{definition}[Movable vertex]\label{def:neighbor_moves}
    Let $n,k \in N^+$,$D \in \Dk$ and $v \in D$.  
    A move of vertex $v$ to another vertex $u \in V(C_n)$ corresponds to a new dominating set $D' \in \Dk$ with $D' = D \setminus \{v\} \cup \{u\}$. 
    We call $v$ counterclockwise movable if $d(v, \cw[D]{v}) \leq 2$. A counterclockwise move moves $v$ to $\ccw{v}$. We call $v$ clockwise movable if $d(v,\ccw[D]{v}) \leq 2$. A clockwise move moves $v$ to $\cw{v}$. We call $v$ free-movable if $d(v,\cw[D]{v}) + d(v,\ccw[D]{v}) \leq 3$.
    A free move can move $v$ to any $u \in V(C_n) \setminus D$. We denote both moves which are not free as neighbor moves.
\end{definition}

See \Cref{fig:movability_cycle} for an example. We note that an arc as described in \Cref{lem:distance_redundant} exists if and only if there exists a freely movable vertex. Therefore, determining movability of vertices and subsequently, the redundancy of a dominating set, only requires knowledge of the distances between vertices and their (counter-)clockwise neighbors. We thus want to capture exactly this information while abstracting other details such as absolute positions of vertices in the dominating set on the cycle graph. For this, we consider the following model:

\begin{definition}[Intermediate Adjacency Graph]\label{def:intermediate_adjacency_graph}
    Let $n,k \in \N^+$ and $D \in \Dk$. Consider the $C_k$ with a bijection $\varphi\colon D \rightarrow V(C_k)$ such that for $u,v \in D: \{\varphi(u), \varphi(v)\} \in E(C_k) \Leftrightarrow v = \cw[D]{u} \land u = \ccw[D]{v}$, i.e., a vertex which corresponds to one in the dominating set is adjacent to those which correspond to its counter- and clockwise neighbors. We further define \[w\colon E(C_k) \rightarrow \{1,2,3\}, \{\varphi(u), \varphi(v)\} \mapsto d(u,v).\] We call $(C_k, w)$ the intermediate adjacency graph of $D$.
\end{definition}

We note that we can define the codomain of the weight function to be $\{1,2,3\}$ since for distinct vertices $u,v \in V(C_n)$ we have $d(u,v) \geq 1$ and when for $v \in D\colon d(v, \ccw[D]{v}) > 3$ or $d(v,\cw[D]{v}) > 3$ there would be a vertex in between these two which would not be dominated, and thus $D \not\in \Dk$.

We now have a graph whose vertices correspond to those in the dominating set with edge weights that represent the distances between them.  However, we will instead consider the edge-to-vertex dual graph of this graph, i.e., the graph that has vertices for edges and edges for vertices. We adapt its definition from \citeauthor{diestel_graph_2025}~\cite{diestel_graph_2025} \emph{(Section 1.1, Line Graph)}.

\begin{definition}[Edge-to-Vertex Dual Graph]
    Let $G = (V,E,w)$. The edge-to-vertex dual graph $EVD(G)$ is the graph on $E$ in which $x,y \in E$ are adjacent as vertices if and only if they are adjacent as edges in $V$. The edge weight function $w$ becomes a function on the vertices of $EVD(G)$.
\end{definition}

\begin{definition}[Adjacency Graph]\label{def:adjacency_graph}
    Let $n,k \in \N^+$ and $D \in \Dk$. Let $A$ be the intermediate adjacency graph of $D$ as defined in \Cref{def:intermediate_adjacency_graph}. Then we say $\adj{D} \coloneq EVD(A)$ is the adjacency graph of $D$, where $EVD(\cdot)$ is the edge-to-vertex dual graph. 
\end{definition}

\begin{figure}
    \centering
    \begin{minipage}[t]{.32\textwidth}
        \centering
        \begin{tikzpicture}[line width=0.5pt, node distance=1.3cm]
        \def\r{1.5}
        \def\a{45}
        \def\off{90}
        \node [circle, draw] (A) at ({\r*cos(0*\a+\off)},{ \r*sin(0*\a+\off)}) {$A$};
        \node [circle, draw,fill=hpired] (B) at ({\r*cos(\a+\off)},{ \r*sin(\a+\off)}) {\color{white}{$B$}};
        \node [circle, draw,fill=hpired] (C) at ({\r*cos(2*\a+\off)},{ \r*sin(2*\a+\off)}) {\color{white}{$C$}};
        \node [circle, draw] (D) at ({\r*cos(3*\a+\off)},{ \r*sin(3*\a+\off)}) {$D$};
        \node [circle, draw,fill=hpired] (E) at ({\r*cos(4*\a+\off)},{\r*sin(4*\a+\off)}) {\color{white}{$E$}};
        \node [circle, draw] (F) at ({\r*cos(5*\a+\off)},{ \r*sin(5*\a+\off)}) {$F$};
        \node [circle, draw] (G) at ({\r*cos(6*\a+\off)},{ \r*sin(6*\a+\off)}) {$G$};
        \node [circle, draw,fill=hpired] (H) at ({\r*cos(7*\a+\off)},{ \r*sin(7*\a+\off)}) {\color{white}{$H$}};
        \path (A) edge (B);
        \path (B) edge (C);
        \path (D) edge (C);
        \path (D) edge (E);
        \path (F) edge (E);
        \path (F) edge (G);
        \path (H) edge (G);
        \path (H) edge (A);
        \node at (0,-2.5) {Dominating set $D$};
        \end{tikzpicture}
    \end{minipage}
    \begin{minipage}[t]{.32\textwidth}
        \centering
        \begin{tikzpicture}[line width=0.5pt, node distance=1.3cm]
        \def\r{1.5}
        \def\a{90}
        \def\off{135}
        \node [circle, draw] (A) at ({\r*cos(0*\a+\off)},{ \r*sin(0*\a+\off)}) {$B$};
        \node [circle, draw] (B) at ({\r*cos(\a+\off)},{ \r*sin(\a+\off)}) {$C$};
        \node [circle, draw] (C) at ({\r*cos(2*\a+\off)},{ \r*sin(2*\a+\off)}) {$E$};
        \node [circle, draw] (D) at ({\r*cos(3*\a+\off)},{ \r*sin(3*\a+\off)}) {$H$};

        \path (A) edge node[left]{$1$} (B);
        \path (B) edge node[below]{$2$}(C);
        \path (D) edge node[right]{$3$}(C);
        \path (D) edge node[above]{$2$} (A);

        \node at (0,-2.5){$(C_4,w)$};
        \end{tikzpicture}
    \end{minipage}
    \begin{minipage}[t]{.32\textwidth}
        \centering
        \begin{tikzpicture}[line width=0.5pt, node distance=1.3cm]
        \def\r{1.5}
        \def\a{90}
        \def\off{90}
        \node [circle, draw] (A) at ({\r*cos(0*\a+\off)},{ \r*sin(0*\a+\off)}) {$2$};
        \node [circle, draw] (B) at ({\r*cos(\a+\off)},{ \r*sin(\a+\off)}) {$1$};
        \node [circle, draw] (C) at ({\r*cos(2*\a+\off)},{ \r*sin(2*\a+\off)}) {$2$};
        \node [circle, draw] (D) at ({\r*cos(3*\a+\off)},{ \r*sin(3*\a+\off)}) {$3$};

        \path (A) edge node[above,xshift=-0.2cm]{$B$} (B);
        \path (B) edge node[below,xshift=-0.2cm]{$C$}(C);
        \path (D) edge node[below,xshift=0.2cm]{$E$}(C);
        \path (D) edge node[above,xshift=0.2cm]{$H$} (A);
        \node at (0,-2.5){Adjacency Graph $\adj{D}$};
        \end{tikzpicture}
    \end{minipage}
    \caption{Left: A dominating set $D$ on $C_8$. Middle: The intermediate adjacency graph of $D$ with corresponding weight function $w$. Right: The Edge-to-Vertex Dual Graph of the intermediate adjacency graph, the adjacency graph $\adj{D}$.}
    \label{fig:adjacency_graph_example}
\end{figure}
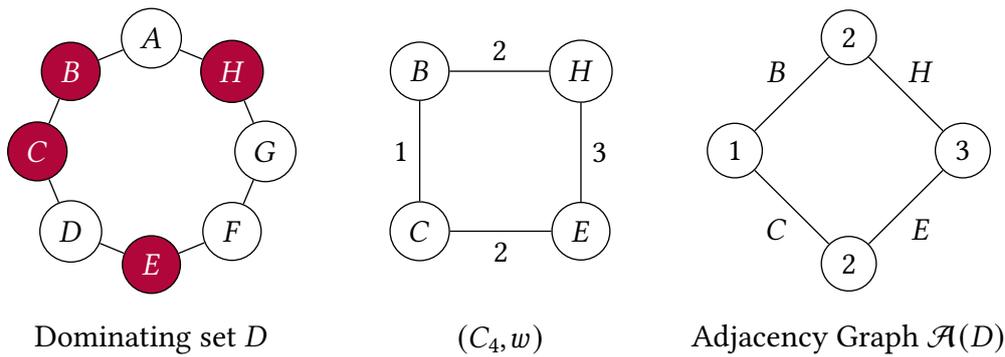

This might seem counterintuitive at first, but this model will be easier to work with in the following analysis. For an example of how the adjacency graph is defined for a dominating set see \Cref{fig:adjacency_graph_example}.
Since we want to observe how the RLS moves vertices around, we often look at sequences of dominating sets that correspond to some sequence of moves. However we would like to limit our analysis to the adjacency graphs of these sets. Therefore, we now define how we derive a sequence of adjacency graphs from a sequence of dominating sets.

\begin{definition}[Adjacency Sequence]\label{def:adjacency_sequence}
    Let $n,k \in \N^+$. Let $(D_i)_{i\in \N}$ be a sequence of dominating sets with $D_i \in \Dk$ for each $i$ where $D_{i+1}$ results from a neighbor move in $D_i$. Consider $A_0 \coloneq \adj{D_0}$ and the corresponding map $\varphi\colon D_0 \rightarrow E(C_k)$. Note that since $D_1$ originated from a neighbor move in $D_0$, the set difference $D_0 \setminus D_1$ only contains the vertex $v$ that was moved and $D_1 \setminus D_0$ the vertex $v'$ it was moved to. 
    Then we define $A_1$ as the adjacency graph with corresponding map \[\varphi\colon D_1 \rightarrow E(C_k), u \mapsto \begin{cases}
        \varphi(v) & \text{if } u = v',\\
        \varphi(u) & \text{else}.
    \end{cases}
    \]
    We naturally extend this definition for all subsequent moves to form the corresponding sequence of adjacency graphs $(A_i)_{i \in \N}$.
\end{definition}

Since the edge corresponding to a moved vertex now simply corresponds to the edge of the vertex it moved to, the only change happens to the vertex weight function.  We will now see how moves of vertices in the dominating set are reflected in these functions. 

Consider for $n,k\in \N^+$ and $D \in \Dk$ a counterclockwise move of a vertex $v$, say to a vertex $v' \coloneq \ccw{v}$. Then this vertex moves closer to its counterclockwise neighbor and further away from its clockwise neighbor. Let $D'$ be the dominating set corresponding to this neighbor move. We have $d(v', \ccw[D']{v'}) = d(v, \ccw[D]{v}) - 1$ and $d(v', \cw[D']{v'}) = d(v, \cw[D]{v'})  + 1$. Let $A \coloneq \adj{D}$ and $A' \coloneq \adj{D'}$ with their respective vertex weight functions $w,w'$ as in \Cref{def:adjacency_sequence}. Let $a,b \in V(C_k)$ be the left and right vertices incident to the edge of $v$ in $A$. Then, for the described counterclockwise move, we have $w'(a) = w(a) -1$ and $w'(b) = w(b) + 1$. Similarly, if we were to consider a clockwise move, we would have $w'(a) = w(a) + 1$ and $w'(b) = w(b) - 1$. This means, that when analyzing a sequence of dominating sets and their adjacency graphs, it actually suffices to look at the corresponding weight functions:

\begin{definition}[Adjacency Weight Sequence]\label{def:adjacency_weight}
    Let $n,k \in \N^+$ and $(D_i \in \Dk)_{i \in \N}$ a sequence of dominating sets. Let $(\adj{D_i} \coloneq (V_i, E_i, w_i))_{i \in \N}$ the sequence of adjacency graphs as defined in \Cref{def:adjacency_sequence}. Then we call the sequence $(w_i)_{i \in \N}$ the adjacency weight sequence with respect to the sequence of dominating sets.
\end{definition}

We provide an example of an adjacency weight sequence of a dominating set sequence as part of \Cref{fig:sequences}.

The following lemma is at its core only a reformulation of \Cref{lem:distance_redundant} in the setting of adjacency graphs, however it is fundamental for the following analysis of the behavior of the RLS.

\begin{lemma}\label{lem:adjacency_redundant}
    Let $n,k \in \N^+$ and $D \in \Dk$. Let $\adj{D} = (V,E,w)$ be an adjacency graph of this dominating set. Then $D$ is not minimal iff there exist an edge $\{u,v\} \in E$ with $w(u) + w(v) \leq 3$.
\end{lemma}
\begin{proof}
    Let $D$ be non minimal. Then, by \Cref{lem:distance_redundant} there exists an arc with length at most $4$ that contains at least $3$ vertices of $D$. Let $a,b,c$ be these vertices. One of them necessarily lies in between the others. Let $b$ be this vertex, let $a = \ccw[D]{b}$ and $c = \cw[D]{b}$. Since the arc has length at most 4, we have $d(a,b) + d(b,c) \leq 3$. Since $a,c$ are (counter-)clockwise neighbors of $b$, by the definition of adjacency graphs, there exist vertices $u,v$ which are adjacent to each other and correspond to these distances. Therefore, there exists an edge $\{u,v\}$ with $w(u) + w(b) = d(a,b) + d(b,c) \leq 3$. 
    \par The proof of the converse follows analogously.
\end{proof}
\begin{corollary}\label{cor:weight_redundant}
    For $D$ and $\adj{D}$ as in \Cref{lem:adjacency_redundant} we have that $D$ is not minimal iff there exists an edge $\{u,v\} \in E$ with either $w(u) = 2 \land w(v) = 1$, $w(u) = 1 \land w(v) = 2$ or $w(u) = 1 \land w(v) = 1$
\end{corollary}
\begin{proof}
    We prove that the existence of an edge $\{u,v\} \in E$ with $w(u) + w(v) \leq 3$ is equivalent to the existence of an edge with the proposed weight configurations. We begin by proving the converse. Suppose an edge with the proposed weight configurations exists. All of these satisfy $w(u) + w(v) \leq 3$. Thus there exists an edge with this total weight restriction. Suppose now that there is an edge $\{u,v\} \in E$ with $w(u) + w(v) \leq 3$. Since the codomain of $w$ is $\{1,2,3\}$, for any $v' \in V$ we have $w(v') \geq 1$. Suppose now $w(u) \geq 2$ and $w(v) \geq 2$. Then $w(u) + w(v) \geq 4$, which is a contradiction. We therefore have to have either $u$ or $v$ with weight $1$. Let \Wlog $w(u) = 1$. Then $w(v) \leq 2$ since otherwise, in the case of $w(v) = 3$, $w(u) + w(v) = 4$, which again is a contradiction. It follows that the proposed weight configurations are the only way for an edge to have total weight less than $3$. The proof is completed by an application of \Cref{lem:adjacency_redundant}.
\end{proof}

\end{section}

\begin{section}{Particle Systems}

With this new definition of non-minimality in the adjacency model, we want to consider mutations of the RLS in light of this model and how they affect the adjacency weight sequence. We want to bound the expected time it takes the RLS to discover dominating sets that correspond to weight functions as described in \Cref{cor:weight_redundant}. For this, we first need to make a few more observations about the model we have just introduced. Consider, for some $n,k \in \N^+$, the cycle graph $C_k = (V,E)$, and a weight function $w\colon V \rightarrow \{1,2,3\}$ corresponding to a dominating set $D \in \Dk$. Suppose we have an edge $\{u,v\} \in E$ with $w(u) = w(v) = 2$. The this edge corresponds to some vertex $a \in D$, $w(u)$ to the distance to its counterclockwise neighbor and $w(v)$ to the distance to its clockwise neighbor. By the definition of movability of vertices, $a$ is clockwise, as well as counterclockwise movable. Consider a clockwise move of $a$. Then the weight function $w'$ corresponding to this move has $w'(u) = w(u) + 1 = 3$ and $w'(v) = w(v) - 1 = 1$. Due to this behavior, we can imagine that a vertex with weight $1$ was \say{composed} of two vertices with weight $2$ as the result of some move. 
\par Now, for a moment, we will imagine a new model where we only really care about vertices with weights $1$ and $2$, since in the end, only those and their locations are relevant to determining whether an underlying dominating set is non-minimal. We will consider these vertices as some particles that roam around the cycle according to some sequence of dominating sets. When two of these particles occupy the same vertex, we say that they met and that this corresponds to a vertex weight of $1$. 
So, in this auxiliary model, to make the underlying dominating set not minimal, we would need two particles to meet and then at least one other to move right beside them. If the behavior of these particles can be modeled as some random walk, we may reduce our analysis to that of some three particles \say{meeting} when randomly walking on the cycle graph. We are now going to formalize this idea.

\begin{definition}[Particle System]
    Let $G = (V,E)$ be a graph and $P$ a set of particles. Let $\psi\colon P \rightarrow V$ be a function that assigns each particle a position on the graph. We denote a particle system by the triple $(G,P,\psi)$.
\end{definition}

\begin{definition}[Adjacency Particle System]\label{def:adjacency_particle_system}
    Let $n,k \in \N^+$ and $D \in \Dk$. Let $\adj{D} = (V,E,w)$ be an adjacency graph of $D$. Let $O \coloneq \{v \in V \with w(v) = 1\}$ and $T \coloneq \{v \in V \with w(v) = 2\}$. Let $P \coloneq \{p_1,p_2,\dots p_{|T| + 2|O|}\}$ be a set of particles. Let $\psi\colon P \rightarrow V$ such that \begin{enumerate}
        \item for $v \in O, |\psi^{-1}(\{v\})| = 2$ and 
        \item for $v \in T, |\psi^{-1}(\{v\})| = 1$.
    \end{enumerate}
    Then the particle system $((V,E), P, \psi)$ is an adjacency particle system of $\adj{D}$.
\end{definition}

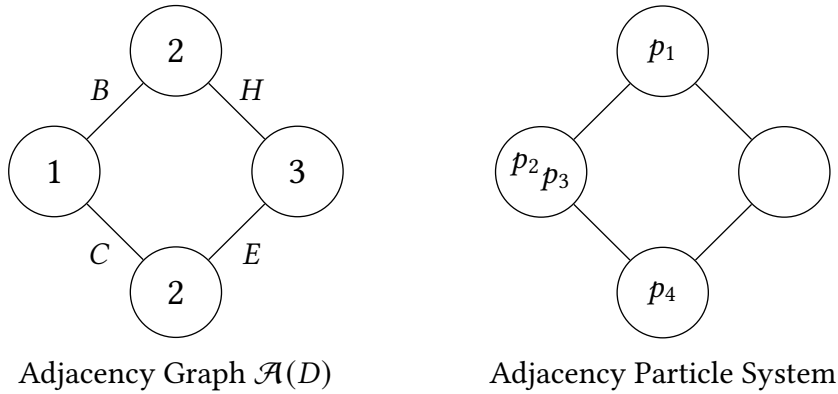
\begin{figure}
    \centering
    \begin{minipage}[t]{.45\textwidth}
        \centering
        \begin{tikzpicture}[line width=0.5pt, node distance=1.3cm]
        \def\r{1.6}
        \def\a{90}
        \def\off{90}
        \node [circle, draw,minimum size=1.2cm] (A) at ({\r*cos(0*\a+\off)},{ \r*sin(0*\a+\off)}) {\large$2$};
        \node [circle, draw,minimum size=1.2cm] (B) at ({\r*cos(\a+\off)},{ \r*sin(\a+\off)}) {\large$1$};
        \node [circle, draw,minimum size=1.2cm] (C) at ({\r*cos(2*\a+\off)},{ \r*sin(2*\a+\off)}) {\large$2$};
        \node [circle, draw,minimum size=1.2cm] (D) at ({\r*cos(3*\a+\off)},{ \r*sin(3*\a+\off)}) {\large$3$};

        \path (A) edge node[above,xshift=-0.2cm]{$B$} (B);
        \path (B) edge node[below,xshift=-0.2cm]{$C$}(C);
        \path (D) edge node[below,xshift=0.2cm]{$E$}(C);
        \path (D) edge node[above,xshift=0.2cm]{$H$} (A);
        
        \node at (0,-2.7){Adjacency Graph $\adj{D}$};
        \end{tikzpicture}
    \end{minipage}
    \begin{minipage}[t]{.45\textwidth}
        \centering
        \begin{tikzpicture}[line width=0.5pt, node distance=1.3cm]
        \def\r{1.6}
        \def\a{90}
        \def\off{90}
        \node [circle, draw,minimum size=1.2cm] (A) at ({\r*cos(0*\a+\off)},{ \r*sin(0*\a+\off)}) {};
        \node [circle, draw,minimum size=1.2cm] (B) at ({\r*cos(\a+\off)},{ \r*sin(\a+\off)}) {};
        \node [circle, draw,minimum size=1.2cm] (C) at ({\r*cos(2*\a+\off)},{ \r*sin(2*\a+\off)}) {};
        \node [circle, draw,minimum size=1.2cm] (D) at ({\r*cos(3*\a+\off)},{ \r*sin(3*\a+\off)}) {};

        \node  at (A) {$p_1$};
        \node [xshift=-0.2cm,yshift=0.1cm] at (B) {$p_2$};
        \node [xshift=0.2cm,yshift=-0.1cm] at (B) {$p_3$};
        \node  at (C) {$p_4$};
    
        \path (A) edge (B);
        \path (B) edge (C);
        \path (D) edge (C);
        \path (D) edge (A);
        \node at (0,-2.7){Adjacency Particle System};
        \end{tikzpicture}
    \end{minipage}
    \caption{Left: The adjacency graph of the dominating set in \Cref{fig:adjacency_graph_example}. Right: The corresponding adjacency particle system. Since there are two vertices with weight $2$ and one with weight $1$, the system has four particles.}
    \label{fig:particle_system_example}
\end{figure}

Something that will become relevant later is the distance between certain particles. We define what we mean by that here.
\begin{definition}[Particle distance]\label{def:particle_distance}
    Let $(G,P,\psi)$ be a particle system. For $p_1, p_2 \in P$ we define their distance $d(\cdot)$ as 
    \[
        d(p_1, p_2) = d_G(\psi(p_1), \psi(p_2))
    \]
    where $d_G(\cdot)$ is the hop distance in $G$. When the particle system is an adjacency particle system, the underlying graph is a cycle (see \Cref{def:adjacency_graph}). We then define the $(\text{counter-})$ clockwise distances $d_{\textrm{ccw}}(\cdot)$ and $d_{\textrm{cw}}(\cdot)$ as the length of the path such that $\psi(p_1)$ and $\psi(p_2)$ are both endpoints of this path and for all vertices $v_i$ on this path, $v_{i+1} = \ccw{v_i}$ or $\cw{v_i}$ respectively.
\end{definition}

As we hinted before, we can now rephrase \Cref{cor:weight_redundant} in light of this new model. \begin{lemma}\label{lem:particles_redundant}
    Let $n,k \in \N^+$ and $D \in \Dk$. Let $\adj{D} = (V,E,w)$ be an adjacency graph of this dominating set. Let $(C_k,P,\psi)$ be a respective adjacency particle system. Then $D$ is not minimal iff there exists an arc $A$ of length $2$ such that $|\{p \in P \with \psi(p) \in A \}| \geq 3$.
\end{lemma}
\begin{proof}
    An arc of length $2$, i.e., an induced $P_2$, consists of exactly two vertices $x,y$ connected by an edge. By definition of the adjacency particle system, for $v \in V$ with $w(v) = 1$ we have $|\{p \in P \with \psi(p) = v\}| = 2$, for $w(v) = 2, |\{p \in P \with \psi(p) = v\}| = 1$ and for $w(v) = 3$, the set is empty. Those combinations that result in at least $3$ particles in this arc are exactly those of \Cref{cor:weight_redundant}, namely $w(x)=1 \land w(y) = 1, w(x) = 2 \land w(y)=1$ and $w(x)=1 \land w(y) = 2$. Since the existence of such an arc is then, by definition, equivalent to the existence of an edge with these weight restrictions, the Lemma follows.
\end{proof}

\begin{definition}[Adjacency Particle Sequence]\label{def:adjacency_particle_sequence}
    Let $n,k \in \N^+$. Let $(D_i \in \Dk)_{i \in \N}$ be a sequence of dominating sets, where $D_{i+1}$ corresponds to a neighbor move (see \Cref{def:neighbor_moves}) in $D_i$. Let $(w_i)_{i \in \N}$ be the adjacency weight sequence as defined in \Cref{def:adjacency_weight}. Let $(C_k, P,\psi_0)$ be an adjacency particle system of $\adj{D_0} = (V,E, w_0)$. We now consider $w_1$. Let $\{u,v\} \in E$ be the edge, where $w_0(u) \neq w_1(u)$ and $w_0(v) \neq w_1(v)$. Let \Wlog $w_1(u) = w_0(u) + 1$ and $w_1(v) = w_0(v) - 1$. By the definition of the weight sequence and the movability of vertices, we have $w_0(u) \leq 2$ and $w_0(v) \geq 2$. Let $p',p'' \in P$. We distinguish between two possible cases:
    \begin{enumerate}
        \item $w_0(u) = 2$ which implies $\psi_0^{-1}(\{u\}) = \{p'\}$. Then \[
        \psi_1 \colon P \rightarrow V, p \mapsto \begin{cases}
            v & \text{if } p = p', \\
            \psi_0(p) & \text{else.}
        \end{cases}
        \]
        \item $w_0(u) = 1$ which implies $\psi_0^{-1}(\{u\}) = \{p', p''\}$. Then we first flip a fair coin $C \sim \Bern{\frac{1}{2}}$. With this, we define \[
        \psi_1\colon P \rightarrow V, p \mapsto \begin{cases}
            v & \text{if } p = p' \text{ and } C = 0, \\
            v & \text{if } p = p'' \text{ and } C = 1, \\
            \psi_0(p) & \text{else.}
        \end{cases}
        \]
    \end{enumerate}
    Now, for $i \in \N^+$, we define $\psi_i$ with respect to $\psi_{i-1}$ according to these rules. We call $(\psi_i)_{i \in \N}$ the adjacency particle sequence.
\end{definition}

For a complete example, how the respective adjacency weight sequence and adjacency particle sequence of a dominating set sequence are defined, we again refer to \Cref{fig:sequences}.

\input{rls-ds/sequences}

Our aim with this particle system is to model how vertices with weight $2$ interact with other vertices. For the first case, with weight exactly $2$, when its weight is increased, that of its neighbor is decreased. When that weight was $3$, it now becomes $2$. We can say, that the weight $2$ particle moved to its neighbor. Thus, in the first case, we leave all particles apart from this one the same and assign it a new position, namely that of its neighbor. In the case that the weight of its neighbor was $2$, then it has now become $1$. Since it had weight $2$ before, there was already a particle on it. Since we want to model weight $1$ vertices with two particles, we again move the particle to this new position. 
In the other case, a vertex with weight $1$ changes its weight. Since it can only increase by one, it becomes $2$. By definition of the adjacency particle system, we have to remove a particle from this vertex. We choose which particle to move with a fair coin flip, as we want the particles to move freely. A deterministic choice, like always moving one particle to the clockwise and the other to the counterclockwise neighbor would complicate our analysis later on.

We now observe how these particles behave with respect to time.

\begin{lemma}\label{lem:particle_movement}
    Let $n,k \in \N^+$. Let $(D_t \in \Dk)_{t \in \N}$ be a sequence of dominating sets resulting from mutations of the RLS (see \Cref{alg:rls}). Let $(w_t)_{t \in \N}$ be an adjacency weight sequence to this sequence. Let $(\psi_t)_{t \in \N}$ be an adjacency particle sequence with a fixed initial particle system $(C_k, P, \psi_0)$. For $t \in \N$, let $M_t$ denote the event that $D_t$ is minimal. Further, for $p \in P$, let $PCN_t(p)$ denote the event that $\psi_{t+1}(p) = \cw{\psi_t(p)}$ and $PCCN_t(p)$ that $\psi_{t+1}(p) = \ccw{\psi_t(p)}$. Then it holds for $p \in P$ that \[
    \Pr{PCN_t(p)}[M_t, w_t(\psi_t(p)) = 2] =\Pr{PCCN_t(p)}[M_t, w_t(\psi_t(p)) = 2]= \frac{1}{4n}
    \] and
     \[
     \Pr{PCN_t(p)}[M_t, w_t(\psi_t(p)) = 1] =\Pr{PCCN_t(p)}[M_t, w_t(\psi_t(p)) = 1]= \frac{1}{8n}.
     \]
\end{lemma}
\begin{proof}
    We first consider the event that $w_t(\psi_t(p)) = 2$.
    Let $v \coloneq \psi_t(p)$ and \Wlog $u \coloneq \cw{v}$. Since $D_t$ is minimal and $w_t(v) = 2$, we must have $w(u) \geq 2$. Then there are three vertices $a,b,c \in D_t$ such that $b = \cw[D_t]{a}$ and $c = \cw[D_t]{b}$ and further $w_t(v) = d(a,b) = 2$ and $w_t(u) = d(b,c) \geq 2$. 
    We now argue that the event that a particle moves to its clockwise neighbor is exactly the event that the RLS swaps $b$ with its clockwise neighbor in iteration $t + 1$.
    \par Suppose the RLS chooses to swap $b$ with its clockwise neighbor in the underlying cycle graph in iteration $t + 1$. Since $b$ is clockwise movable, this swap succeeds. We further have $w_{t+1}(v) = d(a,\cw{b}) = d(a,b) + 1 = w_t(v) + 1$ and similarly $w_{t+1}(u) = w_t(u) - 1$. Therefore, by definition of the adjacency particle sequence, we have $\psi_{t+1}(p) = u = \cw{v} = \cw{\psi_{t}(p)}$. Now, to have $\psi_{t+1}(p) = \cw{\psi_{t}(p)}$, by the definition of the adjacency particle sequence, we need to have $w_{t+1}(v) = w_t(v) + 1$ and $w_{t+1}(u) = w_t(u) - 1$. Since $b$ is the shared vertex between the distances $w(u)$ and $w(v)$, it has to be swapped with its clockwise neighbor in the underlying cycle graph in order for both distances to change. 
    
    Let $S$ be the event that the RLS chooses to perform a swap and let $BCN$ be the event that vertex $b$ and its clockwise neighbor are swapped in iteration $t+1$. We then have \begin{align*}
    &\Pr{PCN_t(p)}[M_t,w_t(v) = 2]
    = \Pr{BCN}[S,M_t, w_t(v) = 2] \cdot \Pr{S}[M_t,w_t(v) = 2] \\
    &= \frac{1}{2n} \cdot \frac{1}{2} = \frac{1}{4n}.
    \end{align*}
    Exactly the same calculations can be made to confirm the probability of the move to the counterclockwise neighbor. 
    \par We now consider the second event, $w_t(\psi_t(p)) = w_t(v) = 1$. Then we necessarily have $w_t(u) = 3$, since we assume $D_t$ to be minimal. We further have $\psi_t^{-1}(v) = \{p,p'\}$ by definition of the particle system. There are again vertices $a,b,c$ as defined above, now with $w_t(v) = d(a,b) = 1$ and $w_t(u) = d(b,c) = 3$.
    As argued before, $p$ or $p'$ move in iteration $t+1$ if and only if the RLS swaps vertex $b$ with its clockwise neighbor in iteration $t+1$. However, now $p$ moves iff this event occurs and $p$ is chosen to move instead of $p'$. Again, let $S$ and $BCN$ be defined as above. Further, let $P$ be the event that $p$ is chosen to move by the adjacency particle sequence. With this we have 
    \begin{align*}
        &\Pr{PCN_t(p)}[M_t,w_t(v) = 1] = \Pr{BCN, P}[S,M_t,w_t(v) = 1] \cdot \Pr{S}[M_t,w_t(v) = 1]\\
        &= \Pr{P}[BCN,S,M_t,w_t(v) = 1]\cdot \Pr{BCN}[S,M_t,w_t(v) = 1]\cdot \Pr{S}[M_t,w_t(v) = 1] \\
        &= \frac{1}{2} \cdot \frac{1}{2n} \cdot \frac{1}{2} = \frac{1}{8n}.
    \end{align*}
    The result for the counterclockwise neighbor again follows the same way.
\end{proof}

With \Cref{lem:particles_redundant,lem:particle_movement} we have now formalized the intuition developed at the beginning of this section. When looking at a sequence of dominating sets that arises from random mutations by the RLS, we can now view this through the lens of particles randomly moving on a cycle. \Cref{lem:particles_redundant} suggests that in order to determine the redundancy of a dominating set and even how close or far it is from being redundant, one can consider those arcs that are occupied by three or more particles. Since we want to bound the time until the RLS first samples a redundant dominating set, it seems intuitive to observe the behavior of the shortest arc that has this property.

\end{section}

\begin{section}{Random Walks and Runtime Bound}

In the previous section, we discovered how our particle systems behaved with respect to a sequence of dominating set. We saw that each particle essentially performs an independent random walk, but once you group them in arcs, you can determine how these individual changes affect the global property of redundancy. Our goal is to ultimately model this process as a Markov Chain, since this greatly simplifies the analysis. By the Markov property, the transition probabilities in one state may only be dependent on the state itself. Considering the length or other properties of the shortest arc is thus not feasible, as there are many other arcs changing in the underlying adjacency particle system, which can influence the transition probabilities to new states. We therefore introduce a new process which will be easier to analyze.

\begin{definition}[Fixed Arc Process]\label{def:fixed_arc_process}
    Let $n,k \in \N^+$. Let $(D_t \in \Dk)_{t \in \N}$ be a sequence of dominating sets resulting from mutations of the RLS. Let $\stocProc{\psi}$ be an adjacency particle sequence to this sequence of dominating sets. Consider the initial particle system $(C_k, P, \psi_0)$. Fix any three particles $P' \coloneq \{p_1, p_2, p_3\} \subseteq P$. By the definition of particle systems, there is a particle such that the other two occupy different vertices. We call this particle $p_1$. We may assume $\ccw[P']{p_1} = p_2$ and $\cw[P']{p_1} = p_3$. For $t \in \N$ we define $A_t$ as \[
    A_t \coloneq \left( \dccw{p_1}{p_2}, \dcw{p_1}{p_3}\right).
    \] where the distance is defined with respect to $\psi_t$.
    When the underlying particle system changes, we may relabel certain particles as $p_1,p_2$ or $p_3$.  For that, we distinguish three main cases: \begin{enumerate}
        \item $\psi_{t-1}(p_1), \psi_{t-1}(p_2)$ and $\psi_{t-1}(p_3)$ are all pairwise distinct. Then each particle occupies a different vertex. Suppose that each of them is also the only particle at this vertex. Regardless of how any of them move, we keep their labeling.
        Suppose now that another particle, different from the other two, occupies the same vertex as one of our fixed particles. The following definition will also apply in the subsequent cases, if not specified otherwise. Suppose there is another particle $p_4$ with $\psi_{t-1}(p_1) = \psi_{t-1}(p_4)$. If in iteration $t$, either $p_1$ or $p_4$ moves, we call whichever particle moved $p_1$, so we simply \say{swap} the particles out. This applies to $p_2$ and $p_3$ as well, if there are other particles at the same location as them.
        \item $\psi_{t-1}(p_1) \neq \psi_{t-1}(p_2) = \psi_{t-1}(p_3)$. Then $p_2$ and $p_3$ have the same (counter-)clockwise distance to $p_1$. If either $p_2$ or $p_3$ moves, we relabel the particle with the smaller counterclockwise distance as the counterclockwise neighbor $p_2$ and the other as the clockwise neighbor $p_3$ respectively.
        \item $\psi_{t-1}(p_1) = \psi_{t-1}(p_2) \neq \psi_{t-1}(p_3)$. The case where $\psi_{t-1}(p_1) = \psi_{t-1}(p_3) \neq \psi_{t-1}(p_2)$ is handled the same way. Suppose either $p_1$ or $p_2$ move towards $p_3$, i.e., their clockwise distance $\dcw{\cdot}{p_3}$ decreases. Then whichever particle moved is relabeled to $p_1$ and the particle that remained is now $p_3$. If either move away from $p_3$ in the same sense, the particle that moved is relabeled as $p_3$ while the remaining particle is now $p_1$.
    \end{enumerate}
    The process $\stocProc{A}$ is defined with fixed $A_0$ where each transition for $t \in \N^+$ from $A_{t-1}$ to $A_t$ is determined by the rules above. We call it the fixed arc process.
\end{definition}

An example of the fixed arc process with respect to a dominating set sequence is shown in \Cref{fig:sequences}.

For this process to be well defined, the adjacency particle system of the initial dominating set has to contain at least three particles. We only consider sequences that can, at some point, contain a redundant dominating set. By \Cref{lem:particles_redundant}, in the particle system of such a set, there is an arc that contains at least three particles. Since by definition of the adjacency particle sequence, the number of particles never differs from that of the initial system, we can be sure that the initial system contains at least three particles as well. 
 
We now analyze how this process behaves with respect to the underlying sequence of dominating sets. We will ultimately use these insights to prove the following theorem:

\begin{restatable}{theorem}{firstRedundancy}\label{thm:first_redundancy_time}
    Let $n,k \in \N^+$. Let $(D_t \in \Dk)_{t\in \N}$ be a sequence of dominating sets resulting from mutations of the RLS. Let $T \coloneq \minS{t \in \N \with D_t \text{ is redundant}}$. Then \[
        \Ex{T} \leq 4n \cdot \frac{(k+1)(k+2)}{2}  \cdot \left(8\log^2(k) + 6(\log(k)+1) \right)
    \]
\end{restatable}

The first step in this proof is establishing why we can analyze this new, related process to prove bounds for the RLS.

\begin{lemma}\label{lem:fixed_arc_ds_bound}
    Let $n,k \in \N^+$ and $(D_t \in \Dk)_{t \in \N}$ be a sequence of dominating sets resulting from mutations of the RLS. Let $T \coloneq \minS{t \in \N \with D_t \text{ is redundant}}$ as in \Cref{thm:first_redundancy_time}. Let $\stocProc{A}$ be a corresponding fixed arc process as in \Cref{def:fixed_arc_process}. Let $H \coloneq \{(0,1),(1,0),(k-1,0),(k-1,1), (0,k-1),(1,k-1)\}$ and $T_A \coloneq \minS{t \in \N \with A_t \in H}$. Then \[
    \E{T} \leq \E{T_A}.
    \]
\end{lemma}
\begin{proof}
    Consider the time $T_A$. Then $A_{T_A} \in H$. We consider the cases $(0,1), (k-1,0)$ and $(k-1,1)$ as the other three are symmetric. For the case $(0,1)$, this has $p_1$ and $p_2$ at the same vertex and $p_3$ at its clockwise neighbor. That means there is an edge with a total of three particles on its endpoints, thus an arc of length $2$ containing three particles. By \Cref{lem:particles_redundant} this implies that $D_{T_A}$ is redundant. For $(k-1, 0)$ we have $p_1$ and $p_3$ at the same vertex. Since $p_2$ has a counterclockwise distance of $\dccw{p_1}{p_2} = k-1$, the endpoint of this counterclockwise path is the clockwise neighbor of this vertex. This is because the underlying cycle only has $k$ vertices and a path of length $k-1$ also has to contain $k$ vertices. We again have an arc of length $2$ containing three particles. The final case $(k-1,1)$ is quite similar. Again, $p_2$ is at the clockwise neighbor of the vertex $p_1$ occupies. Since $p_3$ also has clockwise distance $1$ to $p_1$, it is at the same vertex as $p_2$, so there is again an arc of length $2$ containing three particles. Since all of these cases imply that $D_{T_A}$ is redundant, we have $T \leq T_A$ and thus $\E{T} \leq \E{T_A}$.
\end{proof}

Of course there might have been three other particles that have met before our fixed three have, but since we only want to establish an upper bound, this will suffice. We may also assume that throughout the trajectory of this process, for $t \leq T_A$ as in \Cref{lem:fixed_arc_ds_bound}, $D_t$ is irredundant, i.e., minimal, as otherwise $T < T_A$ as well.

We can bound the time until the RLS first samples a redundant dominating set by bounding the time until the corresponding fixed arc process hits the set defined in \Cref{lem:fixed_arc_ds_bound}. To do so, we determine how this process behaves over time and what transitions it can undertake. For this, we define the neighborhood of a state of the fixed arc process:

\begin{definition}\label{def:stoch_neighborhood}
    Let $\stocProc{A}$ be a fixed arc process as in \Cref{def:fixed_arc_process}. For $t \in \N$ we define the neighborhood $\stNeigh{A_t}$ as \[
        \stNeigh{A_t} \coloneq \{(a,b) \in \N^2 \with \Pr{A_{t+1} = (a,b)}[A_0,\dots,A_t] > 0\}
    \]
\end{definition}

We note that we can limit our analysis to states of $A_t$ with $|A_t| \leq k$. Suppose $|A_t| > k$, then $\dccw{p_1}{p_2} + \dcw{p_1}{p_3} > k$. Since the cycle $C_k$ only has $k$ edges, the counterclockwise and clockwise paths of $p_1,p_2$ and $p_1,p_3$ have to have common edges. This means that \Wlog $p_3$ would have a smaller counterclockwise distance to $p_1$ than $p_2$. For this to happen, $p_2$ and $p_3$ would have to have been at the same vertex at some time and $p_3$ would have moved counterclockwise towards $p_1$. In this case however, by definition, $p_3$ would have been relabeled as $p_2$, making such states impossible. 

We will consider all possible definitions of the next state as described in \Cref{def:fixed_arc_process}. We  first analyze the movement of the process when in the first case of the process definition, i.e., when all particles are at different vertices. 
\begin{lemma}\label{lem:inner_movement}
    Let $\stocProc{A}$ be a fixed arc process as in \Cref{def:fixed_arc_process}. Let $t \in \N$ and $A_t = (x,y)$ such that $2 \leq |A_t| < k$ and $x,y > 0$. Then
    \[
     \stNeigh{A_t} = \{(x-1, y), (x,y-1), (x+1,y-1), (x-1,y+1),(x+1,y),(x,y+1)\} \] and for $n \in \stNeigh{A_t}$ \[
    \Pr{A_{t+1} = n}[A_0,\dots,A_t] = \Pr{A_{t+1} = n}[A_t] =  \frac{1}{4n}.
    \] 
\end{lemma}
\begin{proof}
    We consider particles $p_1,p_2,p_3$ all at different vertices of the graph. We first consider the cases where both components of the pair $A_t$ change.  Suppose \Wlog $p_1$ moves to its clockwise neighbor. Then it moves along the clockwise path to $p_3$, resulting in a shorter distance to $p_3$. At the same time, the counterclockwise distance to $p_2$ gets larger, resulting in the new state $(x+1, y-1)$. When $p_1$ is the only particle at its vertex, by \Cref{lem:particle_movement} this happens with probability exactly $\frac{1}{4n}$. Suppose there was another particle $p_4$ at this vertex. Then the probability of reaching this new state corresponds to either $p_1$ or $p_4$ moving to the clockwise neighboring vertex, since we relabel either to $p_1$ in that case. Let $O,F$ be the events of the corresponding moves. Then, with the second statement of \Cref{lem:particle_movement}, we have \begin{align*}
        &\Pr{A_{t+1} = (x+1,y-1)}[A_0,\dots,A_t] = \Pr{O \lor F}[A_0,\dots,A_t] \\
        &= \Pr{O}[A_t] + \Pr{F}[A_t] = \frac{1}{8n} + \frac{1}{8n} = \frac{1}{4n}.
    \end{align*}
    The probabilities of the case that the next state is $(x-1,y+1)$ can be determined in the same manner, as only the direction of the needed moves changed. The cases where only one component of $A_{t}$ changes are all quite similar, so we only consider $(x+1,y)$ and $(x-1,y)$. This corresponds to the counterclockwise distance from $p_1$ to $p_2$ in- or decreasing. In case that it gets larger, $p_2$ moved away from $p_1$, i.e., to the neighbor in the counterclockwise direction. In the other case, $p_2$ moved to its neighbor in the clockwise direction. Both moves have probability $\frac{1}{4n}$ by \Cref{lem:particle_movement} when $p_2$ is the only particle on its vertex. By similar analysis as above, if there was another particle on its vertex, due to relabeling, the probability of either change is still $\frac{1}{4n}$.
    As particles can only move to their immediate neighbors and only one of them can move at a time, there are six possible moves of three particles. Therefore the considered cases are the only possible ones and the neighborhood of such $A_t$ is exactly the described set. 
\end{proof}

For the second case, we suppose that $p_2,p_3$ occupy the same vertex. We note that this implies $|A_t| = k$. This is because $\dccw{p_1}{p_2}$ and $\dcw{p_1}{p_3}$ is the number of edges on the respective (counter-)clockwise path. If $|A_t| = k$ then all edges are accounted for and the paths must end at the same vertex. Because of this, there are also the impossible cases $A_t =(0,k)$ and $A_t = (k,0)$. A path of $k$ vertices can have at most $k-1$ edges, thus no induced path on a cycle can have $k$ edges, by definition. If we were to loosen the definition, this also corresponds to all three particles at the same vertex, which is impossible by definition of the adjacency particle system (see \Cref{def:adjacency_particle_system}). All other $k - 1$ cases of $|A_t| = k$ correspond to a proper particle system, where $p_2,p_3$ are at the same vertex which is different from the vertex of $p_1$. Thus the following lemma covers all such cases.

\begin{lemma}\label{lem:k_border_movement}
    Let $\stocProc{A}$ be a fixed arc process as in \Cref{def:fixed_arc_process}. Let $t \in \N$ and $A_t = (x,y)$ such that $|A_t| = k$ and $x,y > 0$. Then
    \[
     \stNeigh{A_t} = \{(x-1, y), (x,y-1), (x+1,y-1), (x-1,y+1)\} \] and for $n \in \stNeigh{A_t}$ \[
    \Pr{A_{t+1} = n}[A_0,\dots,A_t] = \Pr{A_{t+1} = n}[A_t] =  \frac{1}{4n}.
    \] 
\end{lemma}
\begin{proof}
    We have already established why this definition covers all cases where $\psi_t(p_2)= \psi_t(p_3) \neq \psi_t(p_1)$. Again, we first consider the movement of $p_1$. Its movement results in subsequent states $(x+1,y-1)$ or $(x-1,y+1)$. The probabilities can be determined in the exact same way as in \Cref{lem:inner_movement}. We therefore only consider $(x-1,y)$ as $(x,y-1)$ is again symmetric. Since $p_2$ and $p_3$ are at the same vertex, when either moves in any direction, it will be relabeled accordingly. Thus, the state transition to $(x-1,y)$ corresponds to either particle moving to the counterclockwise neighboring vertex. Since two vertices occupy the current vertex, either move has probability $\frac{1}{8n}$, so the probability of any one of them moving and thus the resulting state transition is again $\frac{1}{4n}$. There are again six possible moves of these three particles, however the (counter-)clockwise moves of $p_2$ and $p_3$ result in the same state, so the neighborhood of $A_t$ cannot contain more than four states. 
\end{proof}

The last case has one of $p_2,p_3$ at the same vertex as $p_1$. Thus, one of the distances is equal to $0$. Both cannot be $0$, as this would again have three particles on the same vertex. 

\begin{lemma}\label{lem:zero_border_movement}
    Let $\stocProc{A}$ be a fixed arc process as in \Cref{def:fixed_arc_process}. Let $t \in \N$ and $A_t = (x,y)$ such that $2 \leq |A_t| < k$ and either $x=0$ or $y= 0$. Then, if $x = 0$,
    \[
     \stNeigh{A_t} = \{(x,y-1), (x,y+1), (x+1,y-1), (x+1,y)\} \] and for $n \in \stNeigh{A_t}$ 
     \[
    \Pr{A_{t+1} = n}[A_0,\dots,A_t] = \Pr{A_{t+1} = n}[A_t] =  \frac{1}{4n}.
    \]  If $y = 0$, \[
     \stNeigh{A_t} = \{(x-1,y), (x+1,y), (x-1,y+1), (x,y+1)\} 
     \] with identical probabilities.
\end{lemma}
\begin{proof}
    Suppose \Wlog $x = 0$. This has $p_1$ and $p_2$ at the same vertex. We first look at the movement of $p_3$. If $p_3$ moves clockwise, the clockwise distance to $p_1$ will increase, if it moves counterclockwise, it will decrease. This corresponds to subsequent states $(x,y+1)$ and $(x,y-1)$. As argued in \Cref{lem:inner_movement}, regardless of $p_3$ being the only particle at its vertex, the probability of each state transition is $\frac{1}{4n}$. We now consider the moves of $p_1$ and $p_2$. If either move towards $p_3$, this particle will be relabeled as $p_1$. After this move, the resulting state will be $(1, y - 1)$ as $p_1$ is now closer to $p_3$. Since $x$ was $0$ before, this is the state $(x+1,y-1)$. Similarly, if either particle moves away from $p_3$, it will be relabeled as $p_3$. Then the resulting state is $(1,y)$ as $p_1$ still has the same distance to $p_3$ as before, but now the distance between $p_1$ and $p_2$ increased. As analyzed in \Cref{lem:k_border_movement}, since the two particles occupy the same vertex, either state transition has probability $\frac{1}{4n}$. Again, all of the six possible moves are accounted for, so the neighborhood of $A_t$ is exactly the one proposed. The analysis of $y = 0$ is identical and one can see how the proposed neighborhood arises when changing the meaning of $x$ and $y$.
\end{proof}

We note that the state graph of this fixed arc process resembles a slightly altered Triangle Grid graph (see \Cref{def:triangle_grid}) where the vertices $(0,0)$, $(0,k)$ and $(k,0)$ have been removed as these are the impossible states. We therefore claim to be able to analyze the fixed arc process by analysing a random walk on the Triangle Grid.

\begin{definition}[Lazy Random Walk on the Triangle]\label{def:random_walk_triangle}
    For $n \in \N^+$ and a constant $c \in \R_{\geq 6}$ the  $\frac{1}{c}$-random walk on the $n$-triangle is defined as a stochastic process $\stocProc{X}$ where, for $t \in \N$, $X_t$ takes values in $V(T_n)$ and for $X_t = u$ and $v \in \oNeigh{X_t}$\[
        P(u,v) = \Pr{X_{t+1} = v}[X_t] = \frac{1}{c}.
    \]
    We note $P(v,v) = \Pr{X_{t+1} = X_t}[X_t] = 1 - \sum_{v \in \oNeigh{X_t}} \Pr{X_{t+1} = v}[X_t] = 1 - \frac{|\oNeigh{v}|}{c}$.
\end{definition}

\begin{lemma}\label{clm:triangle_fixed_arc}
    Let $n,k \in \N^+$ and let $\stocProc{A}$ be a fixed arc process for a sequence of dominating sets $(D_t \in \Dk)_{t \in \N}$. Let $H \coloneq \{(0,1),(1,0),(k-1,0),(k-1,1), (0,k-1),(1,k-1)\}$ be the target set as defined in \Cref{lem:fixed_arc_ds_bound}. Let $T_A \coloneq \minS{t \in \N \with A_t \in H}$. Let $\stocProc{X}$ be  a $\frac{1}{4n}$-random walk on the $k$-triangle. Let $H_\triangle \coloneq \{(0,0), (k,0), (0,k)\}$  and $T_\triangle \coloneq \minS{t \in \N \with X_t \in H_\triangle}$. Then it holds that $\E{T_A} \leq \E{T_\triangle}$.
\end{lemma}
\begin{proof} \Cref{lem:inner_movement,lem:k_border_movement,lem:zero_border_movement} show that all states $(x,y)$ not in $H$ have the same neighborhood as the vertex $(x,y)$ in the Triangle Grid graph and also identical transition probabilities as in the $\frac{1}{4n}$-random walk. Thus, as long as $A_t$ has not hit $H$, every move taken by $A_t$ can be taken by $X_t$ as well, with the same probability. For $s \in H$, when $A_{T_A} = X_{T_A} = s$, $X_t$ may continue moving arbitrarily, as we stop the process $\stocProc{A}$ at time $T_A$. We therefore simply follow the structure of the Triangular Grid. Let \Wlog $X_{T_\triangle} = (0,0)$. Since the only adjacent vertices are $(0,1)$ and $(1,0)$, the process $\stocProc{X}$ had to have $X_{T'_\triangle} = (0,1)$ or $X_{T'_\triangle} = (1,0)$, where $T'_\triangle < T_\triangle$. This has $T_A \leq T_\triangle$ and thus $\E{T_A} \leq \E{T_\triangle}$.
\end{proof}

We can now use this to prove \Cref{thm:first_redundancy_time}.
\firstRedundancy*

\begin{proof}
We first note that for some $m \in \N^+$ and $\alpha \in \R_{\geq 6}$, the  $\frac{1}{\alpha}$-random walk on the $m$-Triangle is reversible with respect to the uniform distribution $\pi$ over the vertex set of $T_m$. Recall from \Cref{clm:cardinality_triangle} that $\chi_m = |V(T_m)|$. Then, for $x,y \in V(T_m)$ with $x \sim y$, \[
    \pi(x) P(x,y) = \frac{1}{\chi_m}\cdot \frac{1}{\alpha} = \pi(y)P(y,x)
\] and with $x \not\sim y$ \[
    \pi(x) P(x,y) = 0 = \pi(y) P(y,x).
\]

We thus turn to the study of reversible Markov Chains as networks, as introduced in \Cref{chap:3}.
As demonstrated in \Cref{sec:markov_networks}, to model the $\frac{1}{4n}$-random walk on the $k$-triangle, we can now define a network $(T_{k}, c)$ with conductance function $c$, where for $x,y \in V(T_k)$ and $x \sim y$, $c(x,y) = \pi(x) P(x,y)$. Further, we define $c(x,x) = \pi(x)P(x,x)$. Since $\pi$ is uniform and all transitions which are not loops have probability $\frac{1}{4n}$, the conductances for all non-loop edges are $C \coloneq \frac{1}{\chi_{k}}\cdot\frac{1}{4n}$. This means that $(T_k, c)$ is an even network with uniform resistance $R = \chi_{k}\cdot4n$.

Recall that for the stochastic process $\stocProc{D}$ resulting from mutations of the RLS, $T \coloneq \{t \in \N \with D_t \text{ is minimal}\}$ (see \Cref{thm:first_redundancy_time}). For the lazy random walk $\stocProc{X}$ on the $n$-Triangle (\Cref{def:random_walk_triangle}), we defined $H_\triangle \coloneq \{(0,0),(0,k),(k,0)\}$ and $T_\triangle \coloneq \minS{t \in \N \with X_t \in H_\triangle}$ (\Cref{clm:triangle_fixed_arc}). We now give an upper bound for $\Ex{T_\triangle}$.
    
We note again that the commute time $t_{a \leftrightarrow b}$ between nodes $a$ and $b$ in a network, as defined in \Cref{def:commute_time}, is an upper bound for the expected time to first reach $b$ starting from $a$, as it also includes the expected time to return to $a$. With the commute time identity (\Cref{thm:commute_time_identity}) we therefore have
\begin{align*}  
    \E{T_\triangle}[X_0]
    &\leq \E{\minS{t_{s\leftrightarrow X_0} \with s \in H_\triangle}}[X_0] \\
    &\leq \E{t_{(0,0)\leftrightarrow X_0}}[X_0] \\
    &= t_{(0,0)\leftrightarrow X_0} =c_G\effRes{(0,0)}{X_0} = \effRes{(0,0)}{X_0},
    \intertext{since for every $v \in V(T_k)$, $c(v) = \pi(v) = \frac{1}{\chi_{k}}$, we have $c_G = \chi_{k} \cdot \frac{1}{\chi_{k}} =1$. To bound $\effRes{(0,0)}{X_0}$ we use \Cref{lem:resistance_triangle_origin}:}
    &\leq R \cdot \left(8\log^2(k) + 6(\log(k)+1) \right)
    \intertext{As this holds for any vertex, so especially any $X_0$, this removes the need to condition on $X_0$ and we now have } 
    \Ex{T_\triangle} &\leq R \cdot \left(8\log^2(k) + 6(\log(k)+1) \right) \\
    &= 4n \cdot \chi_{k} \cdot \left(8\log^2(k) + 6(\log(k)+1) \right).
    \intertext{We finalize the bound with \Cref{clm:cardinality_triangle}.} 
    &\leq 4n \cdot \frac{(k+1)(k+2)}{2}  \cdot \left(8\log^2(k) + 6(\log(k)+1) \right).
    \end{align*}
    By \Cref{lem:fixed_arc_ds_bound} and \Cref{clm:triangle_fixed_arc} we know that $\Ex{T} \leq \Ex{T_\triangle}$. This proves the theorem.
\end{proof}

This establishes a bound on the first time the RLS samples a redundant dominating set. Let $D_T$ be the dominating set induced by the state of the RLS at iteration $T$. Then $\redVert{D_T} \neq \emptyset$, so the RLS could randomly flip one of the vertices in this set, resulting in a dominating set with strictly smaller cardinality. To bound the expected time this process takes, we need to bound the probability that the RLS in a redundant state actually flips a redundant vertex before sampling a minimal state again.
\begin{lemma}\label{lem:rls_retries}
    Let $n \in \N^+$ and $(D_t \in \D)_{t \in \N}$ be a sequence of dominating sets resulting from mutations of the RLS. Let $S_1 \coloneq \minS{t \in \N \with D_t \text{ is redundant}}$ and $E_1 \coloneq \minS{t \geq S_1 \with D_t \text{ is irredundant}}$. Let $S_i \coloneq \minS{t \geq E_{i-1} \with D_t \text{ is redundant}}$ and $E_i \coloneq \minS{t \geq S_{i} \with D_t \text{ is irredundant}}$ respectively. If no such $t$ exists, we define $S_i = E_i = +\infty$. For $k \in \N^+$, let $TR_k \coloneq |\{S_i \with i \in \N, |D_{S_i}| = k, S_i < +\infty\}|$. Then, for $k \leq |D_0|$, it holds that \[
        \E{TR_k} \leq 2.
    \] Define $SC_k \coloneq \minS{t \in \N \with |D_t| < k}$ and $I_k \coloneq \{i \in \N \with |D_{E_i}| = k\}$. Then, for $i \in I_k$ we have \[
        \E{E_i - S_i} \leq 4n \text{ and } \E{SC_k - S_{\max\{I\}+1}} \leq 4n.
    \]
\end{lemma}
\begin{proof}
    We consider what happens when the current state of the RLS is a redundant dominating set. Then, with probability $\frac{1}{2}$ each, it decides to flip or to move a vertex. If it decides to flip, it can choose to flip one of the redundant vertices. Otherwise, the flip will be rejected, and the next state would still be a redundant dominating set. Recall \Cref{lem:distance_redundant} and suppose there is only a single arc that meets this requirement. Then, in the worst case, the length of this arc is $4$ and there are exactly three vertices on it. This has one vertex at each endpoint of the arc. If they would be moved, the resulting dominating set would no longer be redundant. Otherwise, this configuration would be left unchanged, and the next state would still be a redundant dominating set. This trial process can be modeled by a Markov Chain $\stocProc{E}$. This has the state space $V \coloneq \{G,F,S,M,B\}$ where the states represent the RLS flipping a redundant vertex $(G)$, the flip choice $(F)$, the start of the trial $(S)$, the move choice $(M)$, and the RLS moving a vertex from the redundant configuration $(B)$. We now analyze the transition probabilities. By definition of the Algorithm, we have $\Pr{E_{t+1} = F}[E_t = S] = \Pr{E_{t+1} = M}[E_t=S] = \frac{1}{2}$. Assuming the RLS chose to flip, there is at least one redundant vertex, so this event has a probability of at least $\frac{1}{n}$. We define $\Pr{E_{t+1} = G}[E_t =F] = \frac{1}{n}$, which can only worsen the probability of a successful trial. If the chosen vertex is not redundant, the flip will be rejected, so $\Pr{E_{t+1}= S}[E_t = F] = 1 - \Pr{E_{t+1} = G}[E_t =F]$. On the other hand, given that the RLS chose to move a vertex and assuming the worst possible scenario, there is a $\frac{1}{n}\cdot \frac{1}{2}$ probability of choosing a vertex on one endpoint of the arc and moving it to its neighbor not on the arc. 
    This has event has probability at most $\frac{1}{n}$ and we define $\Pr{E_{t+1} = B}[E_t = M] = \frac{1}{n}$ as before. Again, in the other case $\Pr{E_{t+1} = S}[E_t = M] = 1 - \Pr{E_{t+1} = B}[E_t = M]$. Since this Markov Chain only models one trial, we define $\Pr{E_{t+1} = G}[E_t=G]= \Pr{E_{t+1} = B}[E_t=B] = 1$ so $B$ and $G$ are absorbing states. 
    
    \begin{figure}
        \centering
        \begin{tikzpicture}[->,>=stealth,shorten >=2pt, line width=0.5pt, node distance=2cm]
        \node [circle, draw] (S) {$S$};
        \node [circle, draw] (M) [right of=S,xshift=0cm]{$M$};
        \node [circle, draw] (F) [left of=S] {$F$};
        \node [circle, draw] (G) [left of=F] {$G$};
        \node [circle, draw] (B) [right of=M] {$B$};
        \path (S) edge[bend left = 10] node[above]{$\frac{1}{2}$}(M);
        \path (S) edge[bend right = 10] node[above]{$\frac{1}{2}$}(F);
        \path (F) edge node[above]{$\frac{1}{n}$} (G);
        \path (M) edge node[above]{$\frac{1}{n}$} (B);
        \path (M) edge[bend left = 20] node[below]{$ 1 - \frac{1}{n}$} (S);
        \path (F) edge[bend right = 20] node[below]{$1 - \frac{1}{n}$} (S);
    \end{tikzpicture}
        \caption{The Trial Markov Chain (self-loops are omitted)}
        \label{fig:trial_markov_chain}
    \end{figure}
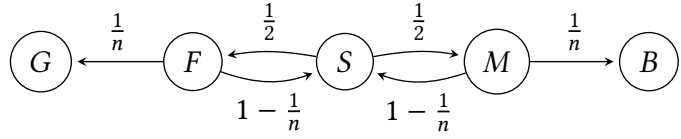
    \newpage
    For a state $v \in V$, we define $p_G(v) \coloneq \Pr{\exists t \in \N: E_t = G}[E_0 = v]$. We are now interested in $p_G(S)$, the probability of the chain ending up in $G$ when started from $S$. By the law of total probability, we have \begin{align}
        \notag p_G(S) &= \Pr{\exists t \in \N: E_t = G}[E_0=S, E_1 = F]\cdot\Pr{E_1 = F}[E_0 = S] \\\notag &~~~~+ \Pr{\exists t \in \N: E_t = G}[E_0=S, E_1 = M]\cdot\Pr{E_1 = M}[E_0 = S] \\\notag
        &= \Pr{\exists t \in \N: E_t = G}[E_0=S, E_1 = F]\cdot\frac{1}{2} \\&~~~~+ \Pr{\exists t \in \N: E_t = G}[E_0=S, E_1 = M]\cdot\frac{1}{2}.\notag
        \intertext{By the Markov property, it is irrelevant whether we started the chain at $S$ and then arrived at $F$ or simply started the chain at $F$:}
        &= \frac{1}{2}\cdot p_G(F) + \frac{1}{2}\cdot p_G(M) \label{eqn:pgs}.
    \end{align}
    In the same manner we can determine $p_G(F)$ and $p_G(M)$ to be \begin{align*}
        p_G(F) &= \frac{1}{n} + \left(1 - \frac{1}{n}\right)\cdot p_G(S) \text{ and} \\
        p_G(M) &= \left(1 - \frac{1}{n} \right)\cdot p_G(S).
    \intertext{
    Substituting these into \eqref{eqn:pgs}, we get }
        p_G(S) &= \frac{1}{2} \cdot \left(\frac{1}{n} + \left(1 - \frac{1}{n}\right)\cdot p_G(S)\right) + \frac{1}{2}\cdot \left(\left(1 - \frac{1}{n} \right)\cdot p_G(S)\right) \\
        &= \frac{1}{2n} + \left(1 - \frac{1}{n} \right)\cdot p_G(S).
        \intertext{By subtracting $\left(1-\frac{1}{n}\right)\cdot p_G(s)$ from both sides and rearranging we have} 
        \frac{1}{2n} &= \left(1 - \left(1 -\frac{1}{n}\right)\right)\cdot p_G(S) = \frac{1}{n}\cdot p_G(S)
        \intertext{This finally yields}
        p_G(S) &= \frac{n}{2n} = \frac{1}{2}.
    \end{align*}
    For $k \leq |D_0|$ and $r \in \N$, $TR_k = r$ if and only if the RLS sampled a redundant dominating set and sampled an irredundant one with same cardinality some time after $r-1$ times and successfully reduced the cardinality in the $r$-th iteration of this process. This can be modeled as $r$ independent runs of our Markov Chain, and $TR_k$ is therefore stochastically dominated by a geometric random variable $G \sim \Geo{\frac{1}{2}}$. This has $\E{TR_k} \leq \E{G} = 2$.

    We bound the expected trial length in a similar fashion. Since we cannot know what state the Markov Chain will end up in, we consider the expected time until the chain is absorbed in either $G$ or $B$. However, now we cannot use an upper bound for the transition probability from $M$ to $B$, since in reality, the chain may be less likely to move to $B$ and rather return to $S$. This would result in shorter paths to an absorbing state being weighted with too much probability mass and thus not yield an upper bound. Therefore, we use the only lower bound we can admit for this transition probability, namely the trivial bound $\Pr{E_{t+1} = B}[E_t = M] \geq 0$. By assuming equality, we isolate $B$, resulting in a restricted set of absorbing states when starting from $S$ and thus in a longer time until the chain is absorbed. For $v \in V$ we now let $f_v \coloneq \E{\minS{t \in \N \with E_t \in \{G,B\}}}[E_0 = v]$ denote the expected time until absorption when starting the chain from $v$. We note that $f_G = f_B = 0$. We now get the system of equations \begin{align}
        f_S &= \frac{1}{2}\left(1 + f_F\right) + \frac{1}{2} \left(1 + f_M\right) \label{eq:fs} \\
        f_F &= \frac{1}{n}(1+f_G) + \left(1 - \frac{1}{n}\right)(1 + f_S) = \frac{1}{n} + \left(1 - \frac{1}{n}\right)(1 + f_S) \label{eq:ff} \\
        f_M &= 1 + f_S \label{eq:fm}
        \intertext{Substituting \eqref{eq:ff} and \eqref{eq:fm} into \eqref{eq:fs}, we get} 
        \notag f_S &= \frac{1}{2}\left(1 + \frac{1}{n} + \left(1 - \frac{1}{n}\right)(1 + f_S)\right) + \frac{1}{2} \left(1 + 1 + f_S\right) \\
        \notag &= \frac{1}{2} + \frac{1}{2}\left(\frac{1}{n} + 1 - \frac{1}{n} + \left(1-\frac{1}{n}\right)f_S\right) + 1 + \frac{f_S}{2} \\
        \notag &= 2 + \frac{1}{2}\left(1-\frac{1}{n}\right)f_S + \frac{f_S}{2} \\
        \notag &= 2 + f_S - \frac{1}{2n}f_S
        \intertext{Solving for $f_S$ we get}
        \notag f_S &= 4n.
    \end{align}
    This is an upper bound for the actual expected time until absorption, which proves the second part of the lemma. 
\end{proof}

We can now use this to bound the expected time the RLS takes when starting with a dominating set of cardinality at most $\floor{\frac{n}{2}}$ to sample a dominating set with smaller cardinality. This can then be used to bound the expected time until the RLS samples an optimal solution, i.e., a minimum dominating set of ~$C_n$.

\begin{theorem}\label{thm:optimal}
    Let $n,k \in \N^+$ with $\ceil{\frac{n}{3}} < k \leq \floor{\frac{n}{2}}$. Let $(D_t \in \D)_{t \in \N}$ be a sequence of dominating sets resulting from mutations of the RLS. Let $t_k \coloneq \minS{t \in \N \with |D_{t}| = k}$ and $T_k \coloneq \minS{t \geq t_k \with |D_t| < k}$. Then \[
        \E{T_k - t_k} \leq 8n \cdot \frac{(k+1)(k+2)+2}{2}  \cdot \left(8\log^2(k) + 6(\log(k)+1) \right)
    \]
\end{theorem}
\begin{proof}
    Let $k \in \N^+$ with the respective restrictions. We are only interested in the time the RLS spends with states inducing dominating sets with cardinality $k$. Since this is independent of the time spent with greater cardinalities, we may assume $|D_0| = k$. Let $T \coloneq \minS{t \in \N \with |D_t| < k}$. Then the states during the interval $[0,T]$ can be classified by the state of their dominating set, namely whether they are redundant or irredundant. Let $R$ be the random variable denoting how often, in this interval, the RLS transitioned from an irredundant state to a redundant one. Let $(TR_i)_{i\in\N}$ and $(TI_i)_{i\in\N}$ be random variables denoting the time spent in redundant states before transitioning back and the time spent in between in irredundant states. For worst case analysis, we may assume that we start in an irredundant state. Then 
    \begin{align*}
        T &= \sum^{R}_{i=1} TR_i + \sum^{R}_{i=1}TI_i \text{ and thus } \\
        \E{T} &= \E{\sum^{R}_{i=1} TR_i} + \E{\sum^{R}_{i=1}TI_i}.
    \end{align*}
    When the RLS switches from a redundant to an irredundant state, we have an arc of length $5$ containing three vertices, since there had to have been an arc of length $4$ containing the same three vertices for the set to have been redundant. We assume states with this characteristic to be sufficiently similar, so that for $i,j > 1$, we may assume $TI_i$ and $TI_j$ to be independent and identically distributed. We apply the same reasoning for all $i,j$ to $TR_i$ and $TR_j$ as they all contain an arc of length $4$ with at least three vertices. We then have
    \begin{align*}
        \E{T} &= \Ex{\sum^{R}_{i=1} TR_i} + \E{TI_1 + \sum^{R-1}_{i=1} TI_{i+1}}.
        \intertext{We can now apply Wald's Equation (\Cref{thm:wald}):}
        &= \Ex{R}\Ex{TR_1} + \Ex{TI_1} + (\Ex{R}-1)\Ex{TI_2}
        \intertext{Using \Cref{lem:rls_retries} and \Cref{thm:first_redundancy_time} we get} 
        &\leq 2\cdot 4n + 2 \cdot 4n \cdot \frac{(k+1)(k+2)}{2}  \cdot \left(8\log^2(k) + 6(\log(k)+1) \right).
        \intertext{Simplifying finally yields}
        \E{T} &\leq 8n \cdot \frac{(k+1)(k+2)+2}{2}  \cdot \left(8\log^2(k) + 6(\log(k)+1) \right).
    \end{align*}
\end{proof}

With this bound, we can now finally prove the main theorem of this thesis.
\begin{theorem}
    Let $n \in \N^+$. It takes the RLS in expectation $\bigO{n^4\log^2(n)}$ iterations to sample an optimal dominating set of $C_n$.
\end{theorem}
\begin{proof}
    Let $(D_t \in \D)_{t \in \N}$ be a sequence of sets resulting from mutations of the RLS. Let $T \coloneq \minS{t \in \N \with |D_t| = \ceil{\frac{n}{3}}}$ be the first time the RLS samples an optimal dominating set. Let $T_F \coloneq \minS{t \geq 0 \with D_t \text{ is a dominating set}}$, $T_H \coloneq \minS{t \geq T_F \with |D_t| \leq \floor{\frac{n}{2}}} - T_F$ and $T_O \coloneq \minS{t \geq T_H \with |D_t| = \ceil{\frac{n}{3}}} - T_H - T_F$. Then $T = T_F + T_H + T_O$. With this we have \begin{align}
        \E{T} = \E{T_F} + \E{T_H} + \E{T_O}. \label{eq:total_time}
    \end{align}
    With \Cref{thm:feasible} we have $\E{T_F} \in \bigO{n\log(n)}$ and with \Cref{thm:half}, $\E{T_H} \in \bigO{n \log(n)}$. Let $t_k \coloneq \minS{t \in \N \with D_t \text{ is a dominating set and } |D_t| = k}$ and $T_k \coloneq \minS{t \in \N \with D_t \text{ is a dominating set and } |D_t| < k}$. Then \begin{align*}
    T_O &= \sum_{k = \ceil{\frac{n}{3}} +1}^{\floor{\frac{n}{2}}} T_k - t_k \text{ and } \\
    \E{T_O} &=  \sum_{k = \ceil{\frac{n}{3}} + 1}^{\floor{\frac{n}{2}}} \E{T_k - t_k}.
    \end{align*}
    By applying \Cref{thm:optimal} we get 
    \begin{align*}
        \E{T_O} &\leq \sum_{k = \ceil{\frac{n}{3}} + 1}^{\floor{\frac{n}{2}}} 8n \cdot \frac{(k+1)(k+2)+2}{2}  \cdot \left(8\log^2(k) + 6(\log(k)+1) \right) \\
        &\leq 4n\cdot\left(8\log^2(n) + 6(\log(n)+1) \right) \cdot \sum_{k = \ceil{\frac{n}{3}} + 1}^{\floor{\frac{n}{2}}} k^2 + 3k + 2 \\
        &\leq 4n\cdot\left(8\log^2(n) + 6(\log(n)+1) \right)\cdot \left(\frac{n(n+1)(2n+1)}{6} +\frac{3n(n+1)}{2}+2n\right) \\
        &\leq 4n\cdot\left(8\log^2(n) + 6(\log(n)+1) \right)\cdot \left(2n^3+3n^2+n + 3n^2 + 3n + 2n\right) \\
        &= 8n^4\cdot \left(8\log^2(n) + 6(\log(n)+1) \right) \\
        &~~~~+ \left(8\log^2(n) + 6(\log(n)+1) \right) \cdot (6n^2+6n) \\
        &\in \bigO{n^4\log^2(n)}.
    \end{align*}
    Thus, for $\E{T}$ substituting into \eqref{eq:total_time}, we have \[
        \E{T} \in \bigO{n\log(n)} + \bigO{n\log(n)} + \bigO{n^4\log^2(n)}.
    \]
    This proves the theorem.
\end{proof}

\end{section}


    \makeatletter
        \def\toclevel@chapter{-1}
        \def\toclevel@section{0}
    \makeatother

    \chapter{Conclusions \& Outlook}

In this thesis, we proved that our Random Local Search heuristic finds a minimum dominating set on a cycle graph $C_n$ in expected time $\bigO{n^4\log^2(n)}$. 
For this, we proved new bounds for the effective resistance in a Triangle Grid graph with uniform resistances.
We determined that a randomly initialized set is most likely not dominating. Starting there, we showed that in expected time $\bigO{n\log(n)}$ RLS would first sample a valid dominating set. We then analyzed properties of dominating sets with $\floor{\frac{n}{2}} + q \leq n$ vertices. We discovered that in such sets there are at least $q$ vertices that can be removed from the set while still maintaining domination. Using this fact in the analysis of the RLS, we showed that it takes expected time $\bigO{n\log(n)}$ to sample a dominating set with $\floor{\frac{n}{2}}$ vertices, which is already only $1.5$ times the optimal size of $\ceil{\frac{n}{3}}$. After this, we introduced the adjacency graph and respective sequence as well as the corresponding weight sequence. These were an intermediate result to capture the dynamic of the RLS swapping vertices. We formalized particle systems based on these models. They allowed us to see the highly localized process of vertex swaps as independent events. Using these, we performed a level-wise analysis of the process to bound the expected time until an optimal solution is found.

We do not believe this bound to be tight for several reasons. With an inequality due to Nash-Williams (given as \emph{Proposition 9.16} in \cite{levin_markov_2017}) one can obtain a lower bound of $\Omega(\log(n))$ for the effective resistance between two vertices in a Triangle Grid graph. We believe that this bound can also be achieved, removing a factor of $\log(n)$ from the runtime of the algorithm. Since the effective resistance is proportional to the escape probability of two nodes, intuitively it seems that this should not be bigger in a triangle than in a square with the same side length $n$, where the resistance is proven to be in $\Theta(\log(n))$. (\Cref{thm:resistance_square}).

The fixed arc process (see \Cref{def:fixed_arc_process}) is in itself very pessimistic. As we noted before, considering only the arc with shortest length at any time could yield better results, however, we were not able to properly analyze this process. The Triangle graph resulting from this process would also be bounded not by $k$, but by the maximum minimum length of an arc containing three particles. However, in the last level, when the current dominating set has cardinality $\frac{n}{3} +1$ and there are three particles, this is roughly $\frac{n}{9}$. It is thus unclear whether this would make an asymptotic difference when summing over all levels. 

Finally, some simulations we conducted suggest that the obtained bound is really too pessimistic. Another way to tackle this analysis that could yield tighter results would be to look at the fitness landscape of this process. When enumerating dominating sets for small $n$ and $k$, we found that the fraction of redundant dominating sets was quite high, only really dropping one level before the optimum. This could be used to analyze the process as a random walk on the state space and to estimate the time until a redundant set is encountered. Since this analysis would look at the trajectory of this process directly without going through several other models, we expect this to better capture the actual complexity.

Additionally, the graph class of cycles is quite simple in terms of the structure of the dominating sets, and the respective performance of the RLS might already be considered poor. To better analyze the applicability of RLS to the Dominating Set problem, other operators and more complex graph classes need to be considered. Other possible graph classes, where Dominating Set is still easy, could be Trees or Cactus graphs.
For Trees, the problem can be solved with dynamic programming \cite{haynes_fundamentals_2023}, and for cactus graphs, there is a linear time algorithm, which finds dominating sets in cycles as a subroutine \cite{hedetniemi_linear_1986}.

    \pagestyle{plain}

    \renewcommand*{\bibfont}{\small}
    \printbibheading
    \addcontentsline{toc}{chapter}{Bibliography}
    \printbibliography[heading = none]

    \addchap{Declaration of Authorship}
    I hereby declare that this thesis is my own unaided work. All direct or indirect sources used are acknowledged as references.\\[6 ex]

\begin{flushleft}
    Potsdam, \today
    \hspace*{2 em}
    \raisebox{-0.9\baselineskip}
    {
        \begin{tabular}{p{5 cm}}
            \hline
            \centering\footnotesize\printAuthor
        \end{tabular}
    }
\end{flushleft}

\end{document}